\documentclass[submission,copyright,creativecommons]{eptcs}
\providecommand{\event}{AiML 2026} 

\usepackage{aiml26}
\usepackage{iftex}
\usepackage{amsmath}
\usepackage{amssymb}
\usepackage{amsthm}
\usepackage{bussproofs} 
\usepackage{graphicx}
\usepackage{tabularx}
\usepackage{cite}
\usepackage{multicol}
\usepackage[ruled,linesnumbered,noend,lined,commentsnumbered]{algorithm2e}
\usepackage[dvipsnames]{xcolor}

\SetAlFnt{\scriptsize}
\SetAlCapNameFnt{\scriptsize}
\SetAlCapFnt{\scriptsize}
\SetNlSty{}{}{:}
\SetAlgoNlRelativeSize{-1}
\SetVlineSkip{0pt}
\definecolor{keywordcolor}{RGB}{153, 0, 0}
\colorlet{commentcolor}{blue}

\SetCommentSty{mycommfont}
\SetKwComment{Comment}{\%~}{}
\SetKwComment{tcp}{\%~}{}
\SetKwIF{If}{ElseIf}{Else}{\textcolor{keywordcolor}{if}}{\textcolor{keywordcolor}{then}}{\textcolor{keywordcolor}{else if}}{\textcolor{keywordcolor}{else}}{\textcolor{keywordcolor}{end if}}%
\SetKwFor{While}{\textcolor{keywordcolor}{while}}{\textcolor{keywordcolor}{do}}{\textcolor{keywordcolor}{endw}}
\SetKw{Return}{\textcolor{keywordcolor}{return}}

\newcolumntype{Y}{>{\centering\arraybackslash}X}

\ifpdf
\usepackage{underscore}         
\usepackage[T1]{fontenc}
\else
\usepackage{breakurl}    
\fi

\newcommand{\Prop}[0]{\mathsf{P}}
\newcommand{\InqB}[0]{\mathsf{InqB}}
\newcommand{\InqBtensor}[0]{\mathsf{InqB}_{\tensor}}
\newcommand{\PID}[0]{\mathsf{PID}}
\newcommand{\PIDtensor}[0]{\mathsf{PID}_{\tensor}}
\newcommand{\LangTL}[0]{\mathcal{L}}
\newcommand{\LangCL}[0]{\mathcal{L}_{c}}

\newcommand{\Langlog}[1]{\mathcal{L}(#1)}
\newcommand{\dep}[1]{{=}(#1)}
\newcommand{\tensor}[0]{\vee}
\newcommand{\globdis}[0]{\dvee}
\newcommand{\Conn}[0]{C}
\newcommand{\Val}[0]{\mathcal{V}(\Prop)}
\newcommand{\sem}[0]{\models}
\newcommand{\semL}[0]{\models_{\Logic}}
\newcommand{\Logic}[0]{L}
\newcommand{\Hsystem}[1]{\mathcal{H}(#1)}
\newcommand{\Gsystem}[1]{\mathcal{G}(#1)}
\newcommand{\HInqB}[0]{\mathsf{HInqB}}
\newcommand{\HInqBtensor}[0]{\mathsf{HInqB}_{\tensor}}
\newcommand{\HPID}[0]{\mathsf{HPID}}
\newcommand{\HPIDtensor}[0]{\mathsf{HPID}_{\tensor}}
\newcommand{\HilbertAxiom}[1]{\mathrm{(#1)}}
\newcommand{\AxA}[1]{\HilbertAxiom{A#1}}
\newcommand{\AxSplit}[0]{\HilbertAxiom{S}}
\newcommand{\AxDN}[0]{\HilbertAxiom{DN}}
\newcommand{\AxIntro}{\HilbertAxiom{Intro}}
\newcommand{\AxMon}{\HilbertAxiom{Mon}}
\newcommand{\AxElim}{\HilbertAxiom{Elim}}
\newcommand{\AxCom}{\HilbertAxiom{Com}}
\newcommand{\AxDis}{\HilbertAxiom{Dis}}
\newcommand{\AxCons}{\HilbertAxiom{C}}
\newcommand{\GInqB}[0]{\mathsf{GInqB}}
\newcommand{\GInqBtensor}[0]{\mathsf{GInqB}_{\tensor}}
\newcommand{\GPID}[0]{\mathsf{GPID}}
\newcommand{\GPIDtensor}[0]{\mathsf{GPID}_{\tensor}}
\newcommand{\provH}[0]{\vdash}
\newcommand{\provG}[0]{\vdash}
\newcommand{\Svar}[0]{\mathfrak{S}}
\newcommand{\Avar}[0]{\mathfrak{A}}
\newcommand{\LAB}[0]{\mathsf{Lab}}
\newcommand{\LABint}[0]{\mathsf{Lab}_{\cap}}
\newcommand{\lab}[2]{#1:#2}
\newcommand{\relat}[2]{#1\subseteq #2}
\newcommand{\equat}[2]{#1 = #2}
\newcommand{\seq}[0]{\Rightarrow}
\newcommand{\extseq}[0]{\Gamma\Rightarrow\Delta}
\newcommand{\Pow}[1]{\mathcal{P}(#1)}
\newcommand{\myfontsize}[0]{\scriptsize}
\newcommand{\Lrule}[1]{#1_{L}}
\newcommand{\Rrule}[1]{#1_{R}}
\newcommand{\RuleFont}[1]{\mathsf{#1}}
\newcommand{\ax}[0]{\RuleFont{ax}}
\newcommand{\axBot}[0]{\RuleFont{ax}_\bot}
\newcommand{\axEmpty}[0]{\RuleFont{ax}_\emptyset}
\newcommand{\Latom}[0]{\RuleFont{at}_{L}}
\newcommand{\Ratom}[0]{\RuleFont{at}_{R}}
\newcommand{\rf}[0]{\RuleFont{rf}}
\newcommand{\tr}[0]{\RuleFont{tr}}
\newcommand{\ul}[0]{\RuleFont{ul}}
\newcommand{\ur}[0]{\RuleFont{ur}}
\newcommand{\sg}[0]{\RuleFont{sg}}
\newcommand{\cd}[0]{\RuleFont{cd}}
\newcommand{\il}[0]{\RuleFont{il}}
\newcommand{\ir}[0]{\RuleFont{ir}}
\newcommand{\dis}[0]{\RuleFont{dis}}
\newcommand{\fin}[0]{\RuleFont{fin}}
\newcommand{\FL}[0]{\RuleFont{flat}}
\newcommand{\MP}[0]{\RuleFont{mp}}
\newcommand{\auxi}[0]{\RuleFont{aux}_{1}}
\newcommand{\auxii}[0]{\RuleFont{aux}_{2}}
\newcommand{\WL}[0]{\RuleFont{w}_{L}}
\newcommand{\WR}[0]{\RuleFont{w}_{R}}
\newcommand{\CL}[0]{\RuleFont{c}_{L}}
\newcommand{\CR}[0]{\RuleFont{c}_{R}}
\newcommand{\cut}[0]{\RuleFont{cut}}
\newcommand{\A}[0]{\pi}
\newcommand{\B}[0]{\sigma}
\newcommand{\C}[0]{\tau}
\newcommand{\rulegen}{R}
\newcommand{\subst}[2]{(#1/#2)}
\newcommand{\dg}[1]{dg(#1)}
\newcommand{\rk}[1]{rk(#1)}
\newcommand{\unionlab}[1]{\theta(#1)}
\newcommand{\altunlab}[1]{\xi(#1)}
\newcommand{\FLfin}[1]{\FL[#1]}
\newcommand{\Vvar}[0]{\mathfrak{V}}
\newcommand{\Splits}[1]{S(#1)}
\newcommand{\I}[0]{I}
\newcommand{\GTsystem}[1]{\mathcal{G}^{\star}(#1)}
\newcommand{\provGT}[0]{\vdash}
\newcommand{\notprovGT}[0]{\nvdash}
\newcommand{\Tree}[0]{\mathfrak{T}}
\newcommand{\Gdown}[0]{{\downarrow}\Gamma_{B}}
\newcommand{\Ddown}[0]{{\downarrow}\Delta_{B}}
\newcommand{\Ltop}[0]{\mathsf{L}(B)}
\newcommand{\Rtop}[0]{\mathsf{R}(B)}
\newcommand{\ProofSearch}[0]{\textsc{ProofSearch}(\Logic)}
\newcommand{\IB}[0]{I_{B}}

\newlength{\mylength}
\newlength{\mybeforehead}
\newlength{\myafterhead}
\usepackage{pict2e,xfp}
\makeatletter
\newcommand{\dvee}{\mathbin{\mathpalette\d@vee\relax}}
\newcommand{\d@vee}[2]{%
  \begingroup
  \sbox\z@{$\m@th#1\vee$}%
  \setlength{\unitlength}{\ht\z@}%
  \kern0.1\wd\z@
  \begin{picture}(\fpeval{0.4\wd0/\unitlength},1)
  \roundcap\d@vee@thickness{#1}
  \Line(0,1.03)(\fpeval{0.4\wd0/\unitlength},0)
  \end{picture}%
  \kern-0.3\wd\z@\box\z@
  \endgroup
}
\newcommand{\d@vee@thickness}[1]{
  \linethickness{%
    1\fontdimen8
      \ifx#1\displaystyle\textfont\else
      \ifx#1\textstyle\textfont\else
      \ifx#1\scriptstyle\scriptfont\else
      \scriptscriptfont\fi\fi\fi 3
  }%
}
\makeatother

\title{Labelled Sequent Calculi for Propositional Team Logics}
\author{
Fausto Barbero
\institute{University of Helsinki}
\email{}
\and
Marianna Girlando
\institute{University of Amsterdam,\\
University of Southern Denmark}
\email{}
\and
Valentin M{\"u}ller
\institute{University of Bern}
\email{}
\and
Fan Yang
\institute{Utrecht University}
\email{}
\and
}

\newcommand{\titlerunning}{Labelled Sequent Calculi for Team-Based  Logics}
\newcommand{\authorrunning}{F. Barbero, M. Girlando, V. M{\"u}ller \&{} F. Yang}

\hypersetup{
  bookmarksnumbered,
  pdftitle    = {\titlerunning},
  pdfauthor   = {\authorrunning},
  pdfsubject  = {EPTCS},
  pdfkeywords = {}
}

\begin{document}
\maketitle

\begin{abstract}
Team semantics is a general framework where formulas are not interpreted with respect to a single point of evaluation, but with respect to sets of such points. Team semantics is used in dependence logic, to reason about dependencies between variables, and in inquisitive logic, to formalize the meaning of questions. We provide sound and complete labelled sequent calculi for four logics based on team semantics: basic inquisitive logic, propositional intuitionistic dependence logic, and their respective extensions with tensor disjunction. For technical reasons, we restrict ourselves to languages with finitely many propositional atoms. The rules of weakening, contraction and cut are shown to be admissible in each of our calculi. In the last part of the paper, we present terminating proof search procedures for variants of our proof systems, in which labels have a simplified structure.
\end{abstract}

\section{Introduction}\label{sec:introduction}

In this paper, we introduce labelled sequent calculi for a family of propositional logics based on team semantics. \emph{Team semantics} is a generalization of Tarskian semantics where formulas are evaluated over sets of evaluation points (called \emph{teams}) rather than single evaluation points. It was introduced by Hodges~\cite{hodges:1997} to characterize dependence and independence between variables. This framework was later systematically developed in \emph{dependence logic}~\cite{vaananen:2007}. \emph{Inquisitive semantics}~\cite{ciardelli:roelofsen:2011} also uses the framework to formalize the meaning of interrogative sentences in natural language. Team semantics has found applications in diverse areas, such as  database theory (e.g.,~\cite{HannulaKL16}), formal semantics of natural language~\cite{ciardelli:2016,Aloni2022}, social choice theory~\cite{PacuitYan16} and quantum information theory~\cite{CoranderHKPV16,AlbertGradel2022,AbramskyPuljVaan2026}.

Most prominent logics based on team semantics are \emph{downward closed}, i.e., their satisfaction relation is preserved under taking subteams. In this work, we focus on this class of logics, and in particular we consider basic inquisitive logic ($\InqB$, from~\cite{ciardelli:roelofsen:2011}), propositional intuitionistic dependence logic ($\PID$, from ~\cite{yang:vaananen:2016}) and their respective extensions with tensor disjunction $\vee$ ($\InqBtensor$ and $\PIDtensor$). All these logics include operators for intuitionistic implication $\to$ and global (or inquisitive) disjunction $\dvee$. The logics $\PID$ and $\PIDtensor$ additionally include dependence atoms. A Hilbert-style system for the logic $\InqB$ was presented in \cite{ciardelli:roelofsen:2011}, while natural deduction systems for propositional dependence logic (which does not have an implication) and its variants (including $\PID$, $\PIDtensor$ and $\InqBtensor$) were introduced in~\cite{yang:vaananen:2016,ciardelli:2016}. 

While the semantic properties of downward closed team logics are well understood in the literature, their proof theory has not been fully investigated. In particular, designing analytic proof systems for these logics remains a challenging task. All the calculi in the literature are generalizations of the Gentzen-style sequent formalism. For instance, Frittella \emph{et al.}~\cite{frittella2016multi-type} introduced a multi-type display calculus for $\InqB$ in the style of Belnap~\cite{belnap1982display}, and Anttila, Iemhoff and Yang~\cite{anttila2025deep} developed a sequent calculus for $\InqBtensor$ employing a form of deep inference. Both systems enrich the structure of Gentzen sequents. 

An alternative approach is to enrich the language of the calculus instead, which is the approach we take in this paper for our labelled sequent calculi for the logics $\InqB$, $\InqBtensor$, $\PID$ and $\PIDtensor$. The calculi we present for these logics are \emph{modular}, meaning that by adding rules they can be adapted to capture logics built up from different sets of connectives or atoms. For instance, by extending our calculi for $\InqB$ and $\InqBtensor$ with two natural rules for dependence atoms, we obtain calculi for $\PID$ and $\PIDtensor$. The modularity feature sets our proof systems apart from the structured calculi presented in \cite{frittella2016multi-type, anttila2025deep}. 

Labelled sequent calculi for the logic $\InqB$ (which does not include the tensor disjunction $\tensor$) have already been presented by Chen and Ma \cite{chen:ma:2017} and later by M{\"u}ller \cite[Chapter~3]{muller:msc:2023}\footnote{In \cite{muller:msc:2023}, M{\"u}ller also presents labelled sequent calculi for \emph{intuitionistic inquisitive logic} \cite{ciardelli:iemhoff:yang:2020} and \emph{inquisitive Kripke logics} \cite[Chapter~6]{ciardelli:phd:2016}. The calculi for the latter were subsequently extended to various systems of \emph{inquisitive modal logic} (see \cite{muller:2024,muller:2026}).} as well as by Litak and Sano \cite[Remark~4]{litak:sano:2026}. Our calculi for the four logics under consideration have a similar design as the ones presented in \cite{chen:ma:2017,muller:msc:2023}, i.e., we also use labels to represent arbitrary teams, singleton teams as well as certain set-theoretic operations on them. Our labelled calculus for $\InqB$, however, is simpler than the ones from \cite{chen:ma:2017} and \cite[Chapter~3]{muller:msc:2023}, since it does not include labels representing intersections. In fact, we will use such labels only for the logics involving tensor disjunction (i.e., for $\InqBtensor$ and $\PIDtensor$). 

A key technical challenge in designing calculi for logics with tensor disjunction is its different behaviour on classical and general formulas, which makes it difficult to define labelled rules for $\InqBtensor$ and $\PIDtensor$ using the standard techniques. To address this difficulty, we restrict our attention to languages containing only \emph{finitely many} propositional variables. This assumption (which does not affect validity in $\InqBtensor$ or $\PIDtensor$) restricts the semantics to finitely many valuations and teams, thus allowing us to introduce a new order rule ($\fin$) that provides enough labels to describe arbitrary valuations and teams. Thanks to this rule, we can simulate the behaviour of the tensor over classical formulas. 

\looseness=-1 Using syntactic arguments, we prove that all our calculi enjoy cut-admissibility. This property, together with the admissibility of additional rules capturing key semantic properties, allows us to establish the completeness of our proof systems. We then present a family of \emph{terminating} labelled calculi for all the logics under consideration. These calculi are obtained by simplifying the structure of labels in our systems. Their construction crucially relies on the assumption that the language has finitely many propositional variables. Our terminating calculus for $\InqB$ is similar to the ones introduced by Litak and Sano \cite{litak:sano:2026} for basic inquisitive logic and for bounded versions of inquisitive predicate logic. We extend this approach to cover $\InqBtensor$, $\PID$ and $\PIDtensor$, and provide an explicit proof search algorithm for all the calculi. 

The paper is structured as follows. In Section~\ref{sec:preliminaries}, we recall the semantics and axiom systems of the logics we consider, and in  Section~\ref{sec:calculi}, we introduce our labelled calculi. In Section~\ref{sec:properties}, we discuss the structural properties of the calculi, and in Section~\ref{sec:completeness}, we establish their completeness. In Section~\ref{sec:proof:search}, finally, we present an explicit proof search algorithm for terminating variants of our calculi.

\section{Preliminaries}\label{sec:preliminaries}

We start by recalling some basic notions. For further details, see \cite{ciardelli:roelofsen:2011,ciardelli:phd:2016,yang:vaananen:2016}. Throughout this paper, we will assume a non-empty set $\Prop$ of \emph{propositional variables} (or \emph{atoms}). For reasons that will become apparent later on, we assume that $\Prop$ is \emph{finite}.\footnote{Strictly speaking, this assumption is only needed in our treatment of systems involving tensor disjunction (i.e., for the logics $\InqBtensor$ and $\PIDtensor$ introduced below). For simplicity, however, we also make this assumption for logics without tensor.} Elements of $\Prop$ are denoted by the letters $p$ and $q$ (possibly with subscripts). The \emph{language of downward-closed team logic}, notation $\LangTL$, is given by the grammar
\begin{equation*}
\varphi ::= p \mid \bot \mid \dep{p} \mid \varphi\wedge\varphi \mid \varphi\tensor\varphi \mid \varphi\globdis\varphi \mid \varphi\rightarrow\varphi \quad (p\in\Prop).
\end{equation*}
We also define $\neg\varphi := \varphi\rightarrow\bot$ and $\varphi\leftrightarrow\psi := (\varphi\rightarrow\psi)\wedge(\psi\rightarrow\varphi)$. We will refer to $\tensor$, $\globdis$ and $\rightarrow$ as \emph{tensor disjunction}, \emph{global disjunction} and \emph{intuitionistic implication}, respectively. In the literature on dependence logic, expressions of the form $\dep{p}$ are known as \emph{constancy dependence atoms} and used to express that $p$ has a fixed value \cite[p.~564]{yang:vaananen:2016}. Dependence atoms with multiple arguments can be defined from constancy atoms by putting $\dep{q_1,\ldots,q_n,p} := (\dep{q_1}\wedge\ldots\wedge\dep{q_n})\rightarrow \dep{p}$. Intuitively, $\dep{q_1,\ldots,q_n,p}$ expresses that the value of $p$ \emph{depends} on the values of $q_1,\ldots,q_n$. In the literature on inquisitive logic, the global disjunction is also known as \emph{inquisitive disjunction} and used to form alternative questions. Thus, in this framework, $\varphi\globdis\psi$ is interpreted as the question \emph{whether $\varphi$ or $\psi$} \cite[Sect.~2.2]{ciardelli:phd:2016}.  

Formulas of our language are evaluated over \emph{teams} (a.k.a.\ \emph{information states}). A team over $\Prop$ is a \emph{set} of valuations $v:\Prop\rightarrow\{0,1\}$. The set of all valuations is given by $\Val := \{v\mid v:\Prop\rightarrow\{0,1\}\}$. Note that, since $\Prop$ is finite, $\Val$ is also finite, and each team $T\subseteq\Val$ contains at most $|\Val| = 2^{|\Prop|}$ elements.

\begin{definition}[Satisfaction]\label{def:satisfaction}
Let $T\subseteq\Val$ be a team. The \emph{satisfaction relation} $\sem$ is defined as follows: 
\begin{itemize}
\setlength{\itemsep}{0pt}
\setlength{\parskip}{0pt}
\item $T\sem p$ ~~iff~~  $v(p)=1$ for all $v\in T$,
\item $T\sem\bot$ ~~iff~~  $T=\emptyset$,
\item $T\sem\dep{p}$ ~~iff~~  $u(p)=v(p)$ for all $u,v\in T$,
\item $T\sem\varphi\wedge\psi$ ~~iff~~  $T\sem\varphi$ and $T\sem\psi$,
\item $T\sem\varphi\tensor\psi$ ~~iff~~  there are $R,S\subseteq T$ with $T=R\cup S$ such that $R\sem\varphi$ and $S\sem\psi$,
\item $T\sem\varphi\globdis\psi$ ~~iff~~  $T\sem\varphi$ or $T\sem\psi$,
\item $T\sem\varphi\rightarrow\psi$ ~~iff~~  for all $S\subseteq T$, if $S\sem\varphi$, then $S\sem\psi$.
\end{itemize}
\end{definition}

It is easy to show that all $\LangTL$-formulas satisfy the \emph{empty team property} and \emph{downward closure}.
\begin{description}
\setlength{\itemsep}{0pt}
\setlength{\parskip}{0pt}
\item[Downward closure:] For all teams $T, S\subseteq\Val$, if $T\sem\varphi$ and $S\subseteq T$, then $S\sem\varphi$.
\item[Empty team property:] $\emptyset\models\varphi$.
\end{description}
A formula $\varphi\in\LangTL$ is said to be \emph{classical}, if it does not contain any occurrences of the connectives $\globdis$ and $\dep{\cdot}$. Classical formulas will be denoted by the letters $\alpha$, $\beta$, $\gamma$. We also write $\LangCL$ for the set of all classical formulas. Classical formulas satisfy both \emph{union closure} and \emph{flatness} (or \emph{truth-conditionality}). 
\begin{description}
\setlength{\itemsep}{0pt}
\setlength{\parskip}{0pt}
\item[Union closure:] For all teams $T,S\subseteq\Val$, if $S\sem\alpha$ and $T\sem\alpha$, then $S\cup T\sem\alpha$.
\item[Flatness (or truth-conditionality):] For every team $T \subseteq\Val$, $T\sem\alpha$ iff $\{v\}\sem\alpha$ for all $v\in T$.
\end{description}

\begin{table}[t]
\caption{Four team-based logics.}
\label{tab:team:logics}
\centering\footnotesize 
\begin{tabular}{lll}
\hline 
\textbf{Logic} & \textbf{Connectives} & \textbf{Hilbert System}\\
\hline 
Basic propositional inquisitive logic ($\InqB$) & $\bot,\wedge,\globdis,\rightarrow$ & $\HInqB$\\ 
Propositional intuitionistic dependence logic ($\PID$) & $\bot,\wedge,\globdis,\rightarrow,\dep{\cdot}$ & $\HPID$\\ 
Propositional inquisitive logic with tensor ($\InqBtensor$) & $\bot,\wedge,\globdis,\rightarrow,\tensor$ & $\HInqBtensor$\\ 
Propositional intuitionistic dependence logic with tensor ($\PIDtensor$) & $\bot,\wedge,\globdis,\rightarrow,\dep{\cdot},\tensor$ & $\HPIDtensor$\\ 
\hline
\end{tabular}
\end{table}

The first column of Table~\ref{tab:team:logics} contains four team-based logics: basic inquisitive logic ($\InqB$), propositional intuitionistic dependence logic ($\PID$), and their respective extensions with tensor disjunction ($\InqBtensor$ and $\PIDtensor$). The second column specifies the corresponding fragment of the language $\LangTL$. So, for example, the language of $\InqB$ consists of all formulas built up from the variables in $\Prop$ by means of the connectives $\bot,\wedge,\globdis,\rightarrow$, and the language of $\PID$ consists of all formulas built up from the variables by means of the connectives $\bot,\wedge,\globdis,\rightarrow,\dep{\cdot}$. The language associated with one of our four logics $\Logic$ will also be denoted by $\Langlog{\Logic}$. Observe that, in particular, $\Langlog{\PIDtensor}$ coincides with the full language $\LangTL$.  

The last column of Table~\ref{tab:team:logics} contains the names of four Hilbert-style systems. The axiom schemes of these systems are presented in the upper section of Figure~\ref{fig:hilbert:systems}. Note that, in the schemes $\AxSplit$, $\AxDN$, $\AxElim$, we require $\alpha$ to range over classical formulas only. The lower section of Figure~\ref{fig:hilbert:systems} contains the definitions of our Hilbert calculi. For example, $\HInqBtensor$ consists of the basic axioms, the tensor axioms and the rule of modus ponens. The Hilbert system associated with a logic $\Logic$ will also be denoted by $\Hsystem{\Logic}$. 

Let now $\Logic$ be any of the logics from Table~\ref{tab:team:logics} and let $\Gamma\cup\{\varphi\}\subseteq\Langlog{\Logic}$. We write $\Gamma\provH\varphi$ and say that $\varphi$ is \emph{provable} from $\Gamma$ in $\Hsystem{\Logic}$, if there exists a sequence $\psi_1,\ldots,\psi_n$ of formulas from $\Langlog{\Logic}$ such that $\varphi = \psi_n$ and every $\psi_i$ with $1\leq i\leq n$ is either an axiom of $\Hsystem{\Logic}$ or an element of $\Gamma$, or it can be obtained from two formulas occurring earlier in the sequence by modus ponens. Moreover, we write $\Gamma\semL\varphi$ and say that $\varphi$ is \emph{entailed} by $\Gamma$ in $\Logic$, if for every team $T\subseteq\Val$, $T\sem\psi$ for all $\psi\in\Gamma$ implies $T\sem\varphi$.

\begin{figure}
\centering
\resizebox{\textwidth}{!}{
\begin{tabular}{|@{\hskip .5cm}l@{~}l@{\hskip .5cm}l@{~}l@{\hskip .5cm}l@{~}l@{\hskip .5cm}|}
\hline
\multicolumn{6}{|c|}{\rule{0pt}{1.5em}\textbf{Basic Axioms}} \\[.25em]
$\AxA{1}$ & $\varphi\rightarrow (\psi\rightarrow\varphi)$ & 
$\AxA{2}$ & $(\varphi\rightarrow (\psi\rightarrow\chi))\rightarrow ((\varphi\rightarrow\psi)\rightarrow (\varphi\rightarrow\chi))$ &
$\AxA{3}$ & $(\varphi\wedge\psi)\rightarrow \varphi$, $(\varphi\wedge\psi)\rightarrow\psi$ \\
$\AxA{4}$ & $\varphi\rightarrow (\psi\rightarrow (\varphi\wedge\psi))$ &
$\AxA{5}$ & $(\varphi\rightarrow\chi)\rightarrow ((\psi\rightarrow\chi)\rightarrow ((\varphi\globdis\psi)\rightarrow\chi))$ &
$\AxA{6}$ & $\varphi\rightarrow (\varphi\globdis\psi)$, $\psi\rightarrow (\varphi\globdis\psi)$ \\
$\AxA{7}$ & $\bot\rightarrow\varphi$ & 
$\AxSplit
$ & $(\alpha\rightarrow (\varphi\globdis\psi))\rightarrow ((\alpha\rightarrow\varphi)\globdis (\alpha\rightarrow\psi))$ &
$\AxDN
$ & $\neg\neg\alpha\rightarrow\alpha$ \\[1em]
\multicolumn{6}{|c|}{\textbf{Tensor Axioms}}\\[.25em]
$\AxIntro$ & $\varphi\rightarrow(\varphi\tensor\psi)$, $\psi\rightarrow(\varphi\tensor\psi)$ &
$\AxElim
$ & $(\varphi\rightarrow\alpha)\rightarrow ((\psi\rightarrow\alpha)\rightarrow ((\varphi\tensor\psi)\rightarrow\alpha))$ &
$\AxMon$ & $(\varphi\rightarrow\psi)\rightarrow ((\varphi\tensor\chi)\rightarrow (\psi\tensor\chi))$\\
$\AxCom$ & $(\varphi\tensor\psi)\rightarrow(\psi\tensor\varphi)$  &
$\AxDis$ & $(\varphi\tensor(\psi\globdis\chi))\rightarrow ((\varphi\tensor\psi)\globdis(\varphi\tensor\chi))$ & & \\[1em]
\multicolumn{6}{|@{\hskip .5cm}l@{\hskip .5cm}|}{\textbf{Constancy Axiom}~$\AxCons$:~~~$\dep{p}\leftrightarrow(p\globdis\neg p)$\hspace{2em} (for $p\in\Prop$)}\\[1em]
\multicolumn{6}{|@{\hskip .5cm}c@{\hskip .5cm}|}{\textit{Side condition:} In $\AxSplit$, $\AxDN$ and $\AxElim$, we require $\alpha\in\LangCL$ to be a classical formula.
}
\\[.5em]
\hline
\multicolumn{6}{|c|}{\rule{0pt}{1.5em}\textbf{Hilbert-Style Systems}} \\[.5em]
\multicolumn{6}{|@{\hskip .5cm}l@{\hskip .5cm}|}{
\begin{tabular}{l@{~~:~~}l @{\hskip 5em} l@{~~:~~}l}
$\HInqB$ & basic axioms; &
$\HInqBtensor$ & basic axioms and tensor axioms; \\
$\HPID$ & basic axioms and constancy axiom; &
$\HPIDtensor$ & basic axioms, tensor axioms and constancy axiom. \\[.5em]
\multicolumn{4}{l}{All systems also include \emph{modus ponens}: from $\varphi$ and $\varphi\rightarrow\psi$, infer $\psi$.}
\\[.75em]
\end{tabular}}\\
\hline 
\end{tabular}
}
\caption{Hilbert-style axiomatizations.}\label{fig:hilbert:systems}
\end{figure}

\begin{theorem}[Soundness and Completeness]\label{th:completeness:hilbert}
Let $\Logic$ be any of the logics from Table~\ref{tab:team:logics}. For every set of formulas $\Gamma\cup\{\varphi\}\subseteq\Langlog{\Logic}$, it is the case that $\Gamma\provH\varphi$ holds in $\Hsystem{\Logic}$ if and only if $\Gamma\semL\varphi$ holds in $\Logic$. 
\end{theorem}

\begin{proof}
See \cite{ciardelli:phd:2016,ciardelli:roelofsen:2011} for (essentially) the completeness proof for $\HInqB$, \cite{yang:vaananen:2016} for (essentially) that of $\HPID$, and \cite{ciardelli:2016} for (essentially) that of $\HInqBtensor$. Completeness of $\HPIDtensor$ follows directly from these results.
\end{proof}

\section{Labelled Sequent Calculi}\label{sec:calculi}

We will now introduce labelled sequent calculi for each of the logics presented in Table~\ref{tab:team:logics}. The labels used in our calculi are built up from two infinite sets of \emph{team variables}, denoted by $\Svar$ and $\Avar$. The variables from $\Svar$ range over \emph{singleton teams} and the variables from $\Avar$ range over teams of \emph{arbitrary size}. Elements of $\Svar$ are denoted by the letters $u$, $v$, $w$, and elements of $\Avar$ are denoted by the letters $x$, $y$, $z$. 

\begin{definition}[Labels]\label{def:labels}
The sets of labels $\LAB$ and $\LABint$ are generated by the following two grammars:     
\begin{equation*}
\begin{array}{ll}
(\LAB) & \pi ::= v \mid x \mid \emptyset \mid \pi\cup\pi \\
(\LABint) & \pi ::= v \mid x \mid \emptyset \mid \pi\cup\pi \mid \pi\cap\pi  
\end{array}%
\qquad\text{($v\in\Svar$ and $x\in\Avar$).} 
\end{equation*}
\end{definition}

The set of labels $\LAB$ will be used in our calculi for $\InqB$ and $\PID$, and the set of labels $\LABint$ will be used in our calculi for $\InqBtensor$ and $\PIDtensor$. Arbitrarily complex labels from any of the sets $\LAB$ and $\LABint$ will be denoted by the meta-variables $\pi$, $\sigma$, $\tau$. Intuitively, $\pi\cup\sigma$ represents the \emph{union} and $\pi\cap\sigma$ represents the \emph{intersection} of the teams denoted by $\pi$ and $\sigma$. The constant $\emptyset$ stands for the \emph{empty team}. 

\begin{definition}
An \emph{interpretation} is a function $I:\Svar\cup\Avar\rightarrow\Pow{\Val}$ that assigns, to each $v\in\Svar$, a singleton team $I(v)\subseteq\Val$, and to each $x\in\Avar$, an arbitrary team $I(x)\subseteq\Val$. An interpretation $I$ is extended to arbitrary labels by defining $I(\emptyset) := \emptyset$, $I(\pi\cup\sigma) := I(\pi)\cup I(\sigma)$ and $I(\pi\cap\sigma) := I(\pi)\cap I(\sigma)$.     
\end{definition}

A \emph{labelled formula} is an expression of the form $\lab{\pi}{\varphi}$, where $\pi$ is a label and $\varphi$ is a formula. Labelled formulas internalise the satisfaction relation into our calculi, so $\lab{\pi}{\varphi}$ may be read as ``The team $\pi$ satisfies the formula $\varphi$''. A \emph{relational atom} is an expression of the form $\relat{\pi}{\sigma}$, where $\pi$ and $\sigma$ are labels, meaning that $\pi$ is a subset of $\sigma$. A \emph{sequent} is an expression of the form $\Gamma\seq\Delta$, where $\Gamma$ is a finite multiset of labelled formulas and relational atoms and $\Delta$ is a finite multiset containing only labelled formulas.

We say that a labelled formula $\lab{\pi}{\varphi}$ is \emph{satisfied} by an interpretation $I$ if it holds that $I(\pi)\sem\varphi$. And a relational atom $\relat{\pi}{\sigma}$ is satisfied by $I$ if $I(\pi)\subseteq I(\sigma)$. A sequent $\Gamma\seq\Delta$ is \emph{valid} if for every interpretation $I$, it is the case that, if $I$ satisfies all elements of $\Gamma$, then $I$ satisfies at least one element of $\Delta$.

\begin{figure}[!t] 
\centering\myfontsize 
\begin{tabularx}{\textwidth}{|@{\hskip .5cm}Y@{\hskip .5cm}|}
\hline
\rule{0pt}{.5cm} 
\textbf{Initial Sequents (Axioms)}
\\[\myafterhead] 
\AxiomC{}
\LeftLabel{$\ax$}
\UnaryInfC{$\lab{v}{p}, \extseq, \lab{v}{p}$}
\DisplayProof
\hspace*{3em}
\AxiomC{}
\LeftLabel{$\axBot$}
\UnaryInfC{$\lab{v}{\bot}, \extseq$}
\DisplayProof
\hspace*{3em}
\AxiomC{}
\LeftLabel{$\axEmpty$}
\UnaryInfC{$\relat{v}{\emptyset}, \extseq$}
\DisplayProof
\\[\mybeforehead] 

\textbf{Basic Logical Rules} 
\\[\myafterhead] 
\AxiomC{$\lab{v}{p},\relat{v}{\A}, \lab{\A}{p}, \extseq$}
\LeftLabel{$\Latom$}
\UnaryInfC{$\relat{v}{\A}, \lab{\A}{p}, \extseq$}
\DisplayProof
\hspace*{1em}
\AxiomC{$\relat{v}{\A}, \extseq, \lab{v}{p}$}
\LeftLabel{$\Ratom$}
\RightLabel{$(\ddagger)$}
\UnaryInfC{$\extseq, \lab{\A}{p}$}
\DisplayProof
\hspace*{1em}
\AxiomC{$\lab{v}{\bot},\relat{v}{\A}, \lab{\A}{\bot}, \extseq$}
\LeftLabel{$\Lrule{\bot}$}
\UnaryInfC{$\relat{v}{\A}, \lab{\A}{\bot}, \extseq$}
\DisplayProof
\hspace*{1em}
\AxiomC{$\relat{v}{\A}, \extseq, \lab{v}{\bot}$}
\LeftLabel{$\Rrule{\bot}$}
\RightLabel{$(\ddagger)$}
\UnaryInfC{$\extseq, \lab{\A}{\bot}$}
\DisplayProof
\\[\mylength]
\AxiomC{$\lab{\A}{\varphi},\lab{\A}{\psi}, \extseq$}
\LeftLabel{$\Lrule{\wedge}$}
\UnaryInfC{$\lab{\A}{\varphi\wedge\psi}, \extseq$}
\DisplayProof
\hspace*{1em}
\def\defaultHypSeparation{\hskip .1in}
\AxiomC{$\extseq, \lab{\A}{\varphi}$}
\AxiomC{$\extseq, \lab{\A}{\psi}$}
\LeftLabel{$\Rrule{\wedge}$}
\BinaryInfC{$\extseq, \lab{\A}{\varphi\wedge\psi}$}
\DisplayProof
\def\defaultHypSeparation{\hskip .2in}
\hspace*{1em}
\def\defaultHypSeparation{\hskip .1in}
\AxiomC{$\lab{\A}{\varphi}, \extseq$}
\AxiomC{$\lab{\A}{\psi}, \extseq$}
\LeftLabel{$\Lrule{\globdis}$}
\BinaryInfC{$\lab{\A}{\varphi\globdis\psi}, \extseq$}
\DisplayProof
\def\defaultHypSeparation{\hskip .2in}
\hspace*{1em}
\AxiomC{$\extseq, \lab{\A}{\varphi}, \lab{\A}{\psi}$}
\LeftLabel{$\Rrule{\globdis}$}
\UnaryInfC{$\extseq, \lab{\A}{\varphi\globdis\psi}$}
\DisplayProof
\\[\mylength]
\AxiomC{$\relat{\A}{\B}, \lab{\B}{\varphi\rightarrow\psi}, \extseq, \lab{\A}{\varphi}$}
\AxiomC{$\relat{\A}{\B}, \lab{\B}{\varphi\rightarrow\psi}, \lab{\A}{\psi}, \extseq$}
\LeftLabel{$\Lrule{\rightarrow}$}
\BinaryInfC{$\relat{\A}{\B}, \lab{\B}{\varphi\rightarrow\psi}, \extseq$}
\DisplayProof
\hspace*{2em}
\AxiomC{$\relat{x}{\A}, \lab{x}{\varphi}, \extseq, \lab{x}{\psi}$}
\LeftLabel{$\Rrule{\rightarrow}$}
\RightLabel{$(\S)$}
\UnaryInfC{$\extseq, \lab{\A}{\varphi\rightarrow\psi}$}
\DisplayProof
\\[\mybeforehead] 

\textbf{Tensor Rules} 
\\[\myafterhead] 
\AxiomC{$\relat{\A}{x\cup y},\lab{x}{\varphi},\lab{y}{\psi},\extseq$}
\LeftLabel{$\Lrule{\tensor}$}
\RightLabel{$(\sharp)$}
\UnaryInfC{$\lab{\A}{\varphi\tensor\psi},\extseq$}
\DisplayProof
\hspace*{2em}
\AxiomC{$\relat{\A}{\B\cup\C},\extseq,\lab{\A}{\varphi\tensor\psi},\lab{\B}{\varphi}$}
\AxiomC{$\relat{\A}{\B\cup\C},\extseq,\lab{\A}{\varphi\tensor\psi},\lab{\C}{\psi}$}
\LeftLabel{$\Rrule{\tensor}$}
\BinaryInfC{$\relat{\A}{\B\cup\C},\extseq,\lab{\A}{\varphi\tensor\psi}$}
\DisplayProof
\\[\mybeforehead] 

\textbf{Constancy Rules} 
\\[\myafterhead] 
\def\defaultHypSeparation{\hskip .1in}
\AxiomC{$\relat{u}{\pi}, \relat{v}{\pi}, \lab{u}{p}, \lab{v}{p}, \lab{\pi}{\dep{p}}, \extseq$}
\AxiomC{$\relat{u}{\pi}, \relat{v}{\pi}, \lab{\pi}{\dep{p}}, \extseq, \lab{u}{p}, \lab{v}{p}$}
\LeftLabel{$\Lrule{=}$}
\BinaryInfC{$\relat{u}{\pi}, \relat{v}{\pi}, \lab{\pi}{\dep{p}}, \extseq$}
\DisplayProof
\def\defaultHypSeparation{\hskip .2in}
\hspace*{0em}
\AxiomC{$\relat{u}{\pi}, \relat{v}{\pi}, \lab{u}{p}, \extseq, \lab{v}{p}$}
\LeftLabel{$\Rrule{=}$}
\RightLabel{$(\ast)$}
\UnaryInfC{$\extseq, \lab{\pi}{\dep{p}}$}
\DisplayProof
\\[\mybeforehead] 

\textbf{Basic Order Rules} 
\\[\myafterhead] 
\AxiomC{$\relat{\A}{\C}, \relat{\A}{\B}, \relat{\B}{\C}, \extseq$}
\LeftLabel{$\tr$}
\UnaryInfC{$\relat{\A}{\B}, \relat{\B}{\C}, \extseq$}
\DisplayProof
\hspace*{6pt}
\AxiomC{$\relat{\A}{\emptyset}, \relat{\A}{v}, \extseq$}
\AxiomC{$\relat{v}{\A}, \relat{\A}{v}, \extseq$}
\LeftLabel{$\sg$}
\BinaryInfC{$\relat{\A}{v}, \extseq$}
\DisplayProof
\hspace*{6pt}
\AxiomC{$\relat{\A_i}{\A_1\cup\A_2}, \extseq$}
\LeftLabel{$\ur$}
\RightLabel{$(i=1,2)$}
\UnaryInfC{$\extseq$}
\DisplayProof 
\\[\mylength] 
\AxiomC{$\relat{v}{\A}, \relat{v}{\A\cup \B}, \extseq$}
\AxiomC{$\relat{v}{\B}, \relat{v}{\A\cup \B}, \extseq$}
\LeftLabel{$\cd$}
\def\defaultHypSeparation{\hskip 6pt}
\BinaryInfC{$\relat{v}{\A\cup \B}, \extseq$}
\def\defaultHypSeparation{\hskip .2in}
\DisplayProof
\hspace*{6pt}
\AxiomC{$\relat{\A\cup\B}{\C}, \relat{\A}{\C}, \relat{\B}{\C}, \extseq$}
\LeftLabel{$\ul$}
\UnaryInfC{$\relat{\A}{\C}, \relat{\B}{\C}, \extseq$}
\DisplayProof%
\hspace*{6pt}
\AxiomC{$\relat{\A}{\A}, \extseq$}
\LeftLabel{$\rf$}
\UnaryInfC{$\extseq$}
\DisplayProof%
\\[\mybeforehead] 

\textbf{Special Order Rules} 
\\[\myafterhead] 
\AxiomC{$\relat{\A}{(\A\cap\B)\cup(\A\cap\C)}, \relat{\A}{\B\cup\C}, \extseq$}
\LeftLabel{$\dis$}
\UnaryInfC{$\relat{\A}{\B\cup\C}, \extseq$}
\DisplayProof
\hspace*{6pt}
\AxiomC{$\relat{\A_1\cap\A_2}{\A_i}, \extseq$}
\LeftLabel{$\il$}
\RightLabel{$(i=1,2)$}
\UnaryInfC{$\extseq$}
\DisplayProof
\hspace*{6pt}
\AxiomC{$\relat{\A}{\B\cap\C}, \relat{\A}{\B}, \relat{\A}{\C}, \extseq$}
\LeftLabel{$\ir$}
\UnaryInfC{$\relat{\A}{\B}, \relat{\A}{\C}, \extseq$}
\DisplayProof
\\[\mylength] 
\AxiomC{$\equat{\A}{\emptyset}, \extseq$}
\AxiomC{$\equat{\A}{v_{1}\cup\ldots\cup v_{n}}, \extseq$}
\LeftLabel{$\fin$}
\RightLabel{$(\Join)$, for $n=2^{|\Prop|}$}
\BinaryInfC{$\extseq$}
\DisplayProof
\\[\mylength] 

\begin{minipage}[b]{\linewidth}
\scriptsize
\emph{Side conditions:} The symbol $(\ddagger)$ indicates that $v\in\Svar$ has to be fresh and $\A$ has to satisfy $\A\notin\Svar$. The symbol $(\S)$ indicates that $x\in \Avar$ must be fresh. In the rule marked with $(\sharp)$, we require $x,y\in\Avar$ to be fresh and distinct. The symbol $(\ast)$ indicates that $u,v\in\Svar$ have to be fresh and distinct. In the rule marked with $(\Join)$, we require $v_{1},\ldots, v_{n}$ to be a collection of $n = 2^{|\Prop|}$ pairwise distinct fresh variables from $\Svar$. Moreover, the notation ``$\equat{\B}{\C}$'' is used as an abbreviation for the pair of relational atoms ``$\relat{\B}{\C},\relat{\C}{\B}$''.  
\strut\end{minipage} \\[.5\mylength]
\hline

\rule{0pt}{.5cm} 
\textbf{Labelled Sequent Calculi} 
\\[\myafterhead] 
\begin{tabular}{l@{~~:~~}l}
\textbullet~~$\GInqB$ & initial sequents, basic logical rules and basic order rules;\\
\textbullet~~$\GPID$ & initial sequents, basic logical rules, constancy rules and basic order rules;\\
\textbullet~~$\GInqBtensor$ & initial sequents and all of the above rules except for the constancy rules;\\
\textbullet~~$\GPIDtensor$ & initial sequents and all of the above rules.\\[.5\mylength] 
\end{tabular}
\\
\hline 
\end{tabularx}
\caption{Labelled sequent calculi.}\label{fig:labelled calculi}
\end{figure}

The axioms and rules of our labelled sequent calculi are presented in the upper section of Figure~\ref{fig:labelled calculi}. Note that, in the rule $\fin$, we use the notation ``$\equat{\B}{\C}$'' as an abbreviation for the pair of atoms ``$\relat{\B}{\C},\relat{\C}{\B}$''. Furthermore, certain rules come with additional side conditions. For example, $v\in\Svar$ has to be a fresh variable in all rules marked with $(\ddagger)$, and $x\in\Avar$ has to be fresh in the rule marked with $(\S)$. Here, a variable is said to be \emph{fresh}, if it does not occur in the conclusion of the respective rule application. The fresh variables introduced by the rules $\Ratom$, $\Rrule{\bot}$, $\Rrule{\rightarrow}$, $\Lrule{\tensor}$, $\Rrule{=}$ and $\fin$ are also referred to as the \emph{eigenvariables} of these rules. 

The rules of our proof systems are inspired by the $\mathsf{G3}$-style approach to labelled calculi as detailed, e.g., in~\cite{negri:2005} for the case of modal logic. However, the richer semantics of team-based logics requires non-trivial adjustments and the addition of several more order rules than in the modal case. The basic logical rules, the tensor rules and the constancy rules simply mirror the satisfaction conditions from Definition~\ref{def:satisfaction}. The basic and special order rules, on the other hand, formalize various set-theoretic properties of teams. The rule $\fin$ plays a somewhat special role here, as it is directly linked to the assumption that $\Prop$ is finite. Intuitively, this rule expresses the fact that every team $\pi$ is either empty or composed of $n = 2^{|\Prop|}$ not necessarily distinct valuations. Note that, in applications of $\fin$, we do not require the union operators in the label $v_{1}\cup\ldots\cup v_{n}$ to be grouped in any specific way. So, for example, in the case $|\Prop| = 2$, the sequents $\equat{\A}{((v_{1}\cup v_{2})\cup v_{3})\cup v_{4}}, \extseq$ and $\equat{\A}{(v_{1}\cup v_{2})\cup (v_{3}\cup v_{4})}, \extseq$ would both be correct instances of the right premise of $\fin$.    

The lower section of Figure~\ref{fig:labelled calculi} contains the definitions of our calculi. For example, the calculus $\GPID$ consists of the initial sequents, the basic logical rules, the constancy rules and the basic order rules. Given any of the logics $\Logic$ from Table~\ref{tab:team:logics}, we will write $\Gsystem{\Logic}$ for the corresponding calculus from Figure~\ref{fig:labelled calculi}.  

\begin{definition}\label{def:provability}
Let $\Logic$ be one of the logics from Table~\ref{tab:team:logics} and let $\Gamma\cup\{\varphi\}\subseteq\Langlog{\Logic}$ be a set of formulas. We will write $\Gamma\provG\varphi$ and say that $\varphi$ is \emph{provable} from $\Gamma$ in the calculus $\Gsystem{\Logic}$, if there exists a finite subset $\Delta\subseteq\Gamma$ and a variable $x\in\Avar$ such that $\lab{x}{\Delta}\seq\lab{x}{\varphi}$ is derivable in $\Gsystem{\Logic}$, where $(\lab{x}{\Delta}) := \{\lab{x}{\psi} \mid \psi\in\Delta\}$.  
\end{definition}

In each of the axioms and rules from Figure~\ref{fig:labelled calculi}, we call $\Gamma$ and $\Delta$ the \emph{left context} and the \emph{right context}, respectively. An occurrence of an expression in an axiom or in the conclusion of a rule is called \emph{principal}, if it does not belong to the context. To ensure that contraction is admissible, we adopt the \emph{closure condition} from \cite{negri:plato:1998,negri:2003,negri:2005,dyckhoff:negri:2012}: if an instance of a rule of one of our calculi contains two principal occurrences of a relational atom $A$ in the conclusion, then also the contracted instance of the rule (in which the two occurrences of $A$ are replaced by a single one) is taken to be part of the corresponding calculus. For example, in the special case where $\pi=\sigma$, we permit not only the instance of the rule $\ul$ shown on the left-hand side below, but also the contracted instance displayed on the right-hand side:  
\begin{center}
{\myfontsize
\AxiomC{$\relat{\pi\cup\pi}{\tau}, \relat{\pi}{\tau}, \relat{\pi}{\tau}, \extseq$}
\LeftLabel{$\ul$}
\UnaryInfC{$\relat{\pi}{\tau}, \relat{\pi}{\tau}, \extseq$}
\DisplayProof}%
\hspace{3cm}
{\myfontsize
\AxiomC{$\relat{\pi\cup\pi}{\tau}, \relat{\pi}{\tau}, \extseq$}
\LeftLabel{$\ul$}
\UnaryInfC{$\relat{\pi}{\tau}, \extseq$}
\DisplayProof}%
\end{center}
The following lemma contains several sequents that are derivable in our calculi. Observe that sequent (\ref{gen:ax:i}) corresponds to \emph{downward closure} and sequents (\ref{gen:ax:ii}) and (\ref{gen:ax:iii}) correspond to the \emph{empty team property}.

\begin{lemma}[Generalized Initial Sequents]\label{lem:gen:ax}
Let $\Logic$ be any of the logics from Table~\ref{tab:team:logics} and let $\Gsystem{\Logic}$ be the corresponding labelled sequent calculus. All sequents of the following forms are derivable in $\Gsystem{\Logic}$:
\setlength\multicolsep{\topsep} 
\begin{multicols}{2}
\begin{enumerate}
\setlength{\itemsep}{0pt}
\setlength{\parskip}{0pt}
\item\label{gen:ax:i} $\relat{\pi}{\sigma}, \lab{\sigma}{\varphi}, \extseq, \lab{\pi}{\varphi}$;
\item\label{gen:ax:ii} $\relat{\pi}{\emptyset},\extseq,\lab{\pi}{\varphi}$;
\item\label{gen:ax:iii} $\relat{\pi}{\sigma}, \lab{\sigma}{\bot}, \extseq, \lab{\pi}{\varphi}$;
\item\label{gen:ax:iv} $\lab{\pi}{\varphi}, \extseq, \lab{\pi}{\varphi}$.
\end{enumerate}
\end{multicols}
\end{lemma}

\begin{proof}
By induction on $\varphi$. Most cases are treated in the same way as in \cite[Lemmas~3.2.1--3.2.2]{muller:msc:2023}. 
\end{proof}

\section{Structural Properties}\label{sec:properties}

We now examine the structural properties of our labelled calculi. Specifically, we will prove the admissibility of weakening and contraction, the invertibility of all rules, and the admissibility of the cut rule. In Section~\ref{sec:completeness}, we will use these results to establish the completeness of our proof systems. In what follows, let $\Logic$ be any of the logics from Table~\ref{tab:team:logics} and let $\Gsystem{\Logic}$ be the associated sequent calculus from Figure~\ref{fig:labelled calculi}.

We first recall some basic definitions. By the \emph{length of a branch} in a derivation, we mean the number of sequents occurring in the branch. The \emph{height of a derivation} $\mathcal{D}$ is the length of one of its longest branches, minus one. A rule $\rulegen$ is \emph{height-preserving admissible} (or \emph{hp-admissible}) in the calculus $\Gsystem{\Logic}$, if whenever each premise of $\rulegen$ is derivable by a derivation of height at most $n$, then the conclusion of $\rulegen$ is derivable by a derivation of height at most $n$. If derivability of the premises implies derivability of the conclusion, but the height is not preserved, $\rulegen$ is just said to be \emph{admissible}. A rule $\rulegen$ is \emph{height-preserving invertible} (or \emph{hp-invertible}) in $\Gsystem{\Logic}$, if whenever the conclusion of $\rulegen$ is derivable by a derivation of height at most $n$, then each premise of $\rulegen$ is derivable, with the same bound on height \cite[Def.~3.4.4]{troelstra:schwichtenberg:1996}. 

Given any variable $s\in\Svar\cup\Avar$ and labels $\pi,\sigma\in\LABint$, we write $\sigma\subst{\pi}{s}$ for the result of substituting $\pi$ for all occurrences of $s$ in $\sigma$. The notation $\Gamma\subst{\pi}{s}$ stands for the result of substituting $\pi$ for $s$ in all labels that occur in a multiset $\Gamma$. By the \emph{substitution rules}, we mean the first two rules displayed in Figure~\ref{fig:structural:rules}.

\begin{figure}[t]
\centering\scriptsize 
\begin{tabularx}{\textwidth}{|@{\hskip .5cm}Y@{\hskip .5cm}|}
\hline 
\rule{0pt}{.6cm} 
\AxiomC{$\Gamma\seq\Delta$}
\LeftLabel{$\subst{u}{v}$}
\RightLabel{$(\ddagger)$}
\UnaryInfC{$\Gamma\subst{u}{v}\seq\Delta\subst{u}{v}$}
\DisplayProof
\hspace{3em}
\AxiomC{$\Gamma\seq\Delta$}
\LeftLabel{$\subst{\pi}{x}$}
\RightLabel{$(\ddagger)$}
\UnaryInfC{$\Gamma\subst{\pi}{x}\seq\Delta\subst{\pi}{x}$}
\DisplayProof
\hspace{3em}
\AxiomC{$\extseq$}
\LeftLabel{$\WL$}
\RightLabel{$(\S)$}
\UnaryInfC{$E,\extseq$} 
\DisplayProof
\hspace{3em}
\AxiomC{$\extseq$}
\LeftLabel{$\WR$}
\UnaryInfC{$\extseq,\lab{\A}{\varphi}$} 
\DisplayProof
\\[\mylength]
\AxiomC{$E,E,\extseq$}
\LeftLabel{$\CL$}
\RightLabel{$(\S)$}
\UnaryInfC{$E,\extseq$}
\DisplayProof
\hspace{5em}
\AxiomC{$\extseq, \lab{\A}{\varphi}, \lab{\A}{\varphi}$}
\LeftLabel{$\CR$}
\UnaryInfC{$\extseq, \lab{\A}{\varphi}$}
\DisplayProof
\hspace{5em}
\AxiomC{$\Gamma\seq\Delta, \lab{\A}{\varphi}$}
\AxiomC{$\lab{\A}{\varphi}, \Pi\seq\Sigma$}
\LeftLabel{$\cut$}
\BinaryInfC{$\Gamma,\Pi\seq \Delta,\Sigma$}
\DisplayProof
\\[\mylength]

\begin{minipage}[b]{\linewidth}
\scriptsize
\emph{Side conditions:} In the rules marked with $(\ddagger)$, we require that $u$ and $v$ are variables from $\Svar$ and that $x$ is a variable from $\Avar$, while $\A$ can be an arbitrarily complex label. The symbol $(\S)$ indicates that $E$ may be either a relational atom or a labelled formula.
\strut\end{minipage}
\\[.5\mylength] 
\hline 
\end{tabularx}
\caption{The substitution rules and the structural rules of weakening, contraction and cut.}\label{fig:structural:rules}
\end{figure}

\begin{lemma}
The substitution rules are hp-admissible in $\Gsystem{\Logic}$.
\end{lemma}

The admissibility of both substitution rules is established simultaneously, by induction on the height of a derivation for the premise. Further details are provided in \cite[Prop.~3.2.6]{muller:msc:2023}. Using this result, we can now prove the admissibility of the weakening rules, the contraction rules and the cut rule (see Figure~\ref{fig:structural:rules}). 

\begin{theorem}[Structural properties]\label{thm:struct}
The calculus $\Gsystem{\Logic}$ has the following structural properties:
\setlength\multicolsep{\topsep}
\begin{multicols}{2}
\begin{enumerate}
\setlength{\itemsep}{0pt}
\setlength{\parskip}{0pt}
\item\label{thm:struct:i} The weakening rules are hp-admissible.
\item\label{thm:struct:ii} All rules of $\Gsystem{\Logic}$ are hp-invertible.
\item\label{thm:struct:iii} The contraction rules are hp-admissible.
\item\label{thm:struct:iv} The cut rule is admissible in $\Gsystem{\Logic}$.
\end{enumerate}
\end{multicols}
\end{theorem}

The proof of this theorem follows the same structure as the standard admissibility proofs for structural rules in $\mathsf{G3}$-style systems (see~\cite{negri:plato:2001,negri:2005}). Further details are in Appendix~\ref{app:proof:struct:properties}, and similar proofs are in \cite{muller:msc:2023}.

\section{Completeness}\label{sec:completeness}

To establish the completeness of our proof systems, we will now prove that each of the labelled sequent calculi from Figure~\ref{fig:labelled calculi} is complete with respect to the corresponding Hilbert-style system from Figure~\ref{fig:hilbert:systems}. For technical reasons, we first prove the admissibility of the rules $\auxi$ and $\auxii$ presented in Figure~\ref{fig:further:admissible:rules}.

\begin{figure}[t]
\centering\scriptsize 
\begin{tabularx}{\textwidth}{|@{\hskip 0pt}Y@{\hskip 0pt}|}
\hline 
\rule{0pt}{.7cm} 
\begin{math}
\AxiomC{$\relat{\A}{(\B\cup\C)\cup(\B'\cup\C')},\extseq$}
\LeftLabel{$\auxi$}
\UnaryInfC{$\relat{\A}{\B\cup\B'},\extseq$}
\DisplayProof
\hspace*{5pt}
\AxiomC{$\relat{\A}{(\C\cup\B)\cup(\C'\cup\B')},\extseq$}
\LeftLabel{$\auxii$}
\UnaryInfC{$\relat{\A}{\B\cup\B'},\extseq$}
\DisplayProof
\hspace*{5pt}
\def\defaultHypSeparation{\hskip .1in}
\AxiomC{$\Gamma\seq \lab{x}{\varphi}$}
\AxiomC{$\Delta\seq \lab{x}{\varphi\rightarrow\psi}$}
\LeftLabel{$\MP$}
\BinaryInfC{$\Gamma,\Delta\seq \lab{x}{\psi}$}
\DisplayProof
\def\defaultHypSeparation{\hskip .2in}
\hspace*{5pt}
\AxiomC{$\relat{v}{\A}, \extseq, \lab{v}{\alpha}$}
\LeftLabel{$\FL$}
\RightLabel{$(\ddagger)$}
\UnaryInfC{$\extseq, \lab{\A}{\alpha}$}
\DisplayProof
\end{math}
\\[1.5em]
$(\ddagger)$ means that $\alpha$ has to be a classical formula and $v$ has to be a fresh variable from $\Svar$.
\\[.5\mylength] 
\hline 
\end{tabularx}
\caption{Further admissible rules.}\label{fig:further:admissible:rules}
\end{figure}

\begin{lemma}\label{lem:auxiliary:rules}
The rules $\auxi$ and $\auxii$ are admissible in each of the labelled calculi from Figure~\ref{fig:labelled calculi}.
\end{lemma}

\looseness=-1 A proof is provided in Appendix~\ref{app:proof:aux:rules}. Next, we need to establish the admissibility of the rules $\MP$ and $\FL$ (see Figure~\ref{fig:further:admissible:rules}). The rule $\MP$ corresponds to \emph{modus ponens} and the rule $\FL$ accounts for the \emph{flatness} of classical formulas: if a classical formula $\alpha$ is true under every valuation $v$ in a team, then it is also satisfied by this team (see Section~\ref{sec:preliminaries}). For the calculi not involving tensor (i.e., for $\GInqB$ and $\GPID$), the admissibility of the rule $\FL$ can simply be established by induction on the structure of $\alpha$ (cf.\ \cite[Prop.~3.2.14]{muller:msc:2023}). For the other two calculi (i.e., for $\GInqBtensor$ and $\GPIDtensor$), we first need the following lemma. 

\begin{lemma}\label{lem:union:closure}
Let $\Logic$ be $\InqBtensor$ or $\PIDtensor$, let $\Gsystem{\Logic}$ be the corresponding labelled calculus and let $\alpha\in\Langlog{\Logic}$ be a classical formula. All sequents of the form $\lab{\pi}{\alpha},\lab{\sigma}{\alpha},\Gamma\seq\Delta,\lab{\pi\cup\sigma}{\alpha}$ are derivable in $\Gsystem{\Logic}$.  
\end{lemma}

\begin{proof}
Let $\Logic$ be one of $\InqBtensor$ and $\PIDtensor$, let $\Gsystem{\Logic}$ be the associated labelled calculus and let $\alpha\in\Langlog{\Logic}$ be classical. The proof is by induction on the structure of $\alpha$. We only consider the following cases.

\emph{Case 1:} $\alpha = p$ for some variable $p\in\Prop$. In this case, the desired sequent can be derived as follows:
\begin{center}
\myfontsize
\AxiomC{Lemma~\ref{lem:gen:ax}~(\ref{gen:ax:i})}
\dashedLine
\UnaryInfC{$\ldots,\relat{v}{\pi},\lab{\pi}{p},\Gamma\seq\Delta,\lab{v}{p}$}
\AxiomC{Lemma~\ref{lem:gen:ax}~(\ref{gen:ax:i})}
\dashedLine
\UnaryInfC{$\ldots,\relat{v}{\sigma},\lab{\sigma}{p},\Gamma\seq\Delta,\lab{v}{p}$}
\LeftLabel{$\cd$}
\BinaryInfC{$\relat{v}{\pi\cup\sigma}, \lab{\pi}{p},\lab{\sigma}{p},\Gamma\seq\Delta,\lab{v}{p}$}
\LeftLabel{$\Ratom$}
\UnaryInfC{$\lab{\pi}{p},\lab{\sigma}{p},\Gamma\seq\Delta,\lab{\pi\cup\sigma}{p}$}
\DisplayProof 
\end{center}

\emph{Case 2:} $\alpha = \beta\tensor\gamma$, where $\beta$ and $\gamma$ are classical. Using the admissible rules $\auxi$ and $\auxii$ from Figure~\ref{fig:further:admissible:rules}, we now construct the following derivation: 
\begin{center}
\myfontsize
\AxiomC{By~ind.~hyp.}
\dashedLine
\UnaryInfC{$\ldots, \lab{x_1}{\beta}, \lab{y_1}{\beta} \seq \ldots, \lab{x_1\cup y_1}{\beta}$\vphantom{$()$}}
\AxiomC{By~ind.~hyp.}
\dashedLine
\UnaryInfC{$\ldots, \lab{x_2}{\gamma}, \lab{y_2}{\gamma} \seq \ldots, \lab{x_2\cup y_2}{\gamma}$\vphantom{$()$}}
\LeftLabel{$\Rrule{\tensor}$}
\BinaryInfC{$\ldots, \relat{\pi\cup\sigma}{(x_1\cup y_1)\cup (x_2\cup y_2)}, \lab{x_1}{\beta}, \lab{x_2}{\gamma}, \lab{y_1}{\beta}, \lab{y_2}{\gamma}, \Gamma\seq\Delta, \lab{\pi\cup\sigma}{\beta\tensor\gamma}$}
\LeftLabel{$\ul$}
\UnaryInfC{$\relat{\pi}{(x_1\cup y_1)\cup (x_2\cup y_2)}, \relat{\sigma}{(x_1\cup y_1)\cup (x_2\cup y_2)}, \lab{x_1}{\beta}, \lab{x_2}{\gamma}, \lab{y_1}{\beta}, \lab{y_2}{\gamma}, \Gamma\seq\Delta, \lab{\pi\cup\sigma}{\beta\tensor\gamma}$}
\LeftLabel{$\auxi$}
\UnaryInfC{$\relat{\pi}{x_1\cup x_2},\relat{\sigma}{(x_1\cup y_1)\cup (x_2\cup y_2)}, \lab{x_1}{\beta}, \lab{x_2}{\gamma}, \lab{y_1}{\beta}, \lab{y_2}{\gamma}, \Gamma\seq\Delta, \lab{\pi\cup\sigma}{\beta\tensor\gamma}$}
\LeftLabel{$\auxii$}
\UnaryInfC{$\relat{\pi}{x_1\cup x_2}, \relat{\sigma}{y_1\cup y_2}, \lab{x_1}{\beta}, \lab{x_2}{\gamma}, \lab{y_1}{\beta}, \lab{y_2}{\gamma}, \Gamma\seq\Delta, \lab{\pi\cup\sigma}{\beta\tensor\gamma}$}
\LeftLabel{$\Lrule{\tensor}$}
\UnaryInfC{$\relat{\pi}{x_1\cup x_2}, \lab{x_1}{\beta}, \lab{x_2}{\gamma}, \lab{\sigma}{\beta\tensor\gamma}, \Gamma\seq\Delta, \lab{\pi\cup\sigma}{\beta\tensor\gamma}$}
\LeftLabel{$\Lrule{\tensor}$}
\UnaryInfC{$\lab{\pi}{\beta\tensor\gamma}, \lab{\sigma}{\beta\tensor\gamma}, \Gamma\seq\Delta, \lab{\pi\cup\sigma}{\beta\tensor\gamma}$}
\DisplayProof
\end{center}

\emph{Case 3:} $\alpha = \beta\rightarrow\gamma$, where $\beta$ and $\gamma$ are classical. We can now construct the derivation  
\begin{center}
\myfontsize
\AxiomC{Lemma~\ref{lem:gen:ax}~(\ref{gen:ax:i})}
\dashedLine
\UnaryInfC{$\relat{x\cap\pi}{x}, \ldots, \lab{x}{\beta} \seq \ldots, \lab{x\cap\pi}{\beta}$}
\LeftLabel{$\il$}
\UnaryInfC{$\ldots, \lab{x}{\beta} \seq \ldots, \lab{x\cap\pi}{\beta}$}
\AxiomC{Lemma~\ref{lem:gen:ax}~(\ref{gen:ax:i})}
\dashedLine
\UnaryInfC{$\relat{x\cap\sigma}{x}, \ldots, \lab{x}{\beta} \seq \ldots, \lab{x\cap\sigma}{\beta}$}
\LeftLabel{$\il$}
\UnaryInfC{$\ldots, \lab{x}{\beta} \seq \ldots, \lab{x\cap\sigma}{\beta}$}
\AxiomC{$\mathcal{D}$}
\noLine
\UnaryInfC{$\ldots, \relat{x}{\pi\cup\sigma}, \lab{x\cap\pi}{\gamma}, \lab{x\cap\sigma}{\gamma} \seq \ldots, \lab{x}{\gamma}$}
\LeftLabel{$\Lrule{\rightarrow}$}
\BinaryInfC{$\ldots, \relat{x\cap\sigma}{\sigma}, \relat{x}{\pi\cup\sigma}, \lab{x}{\beta}, \lab{x\cap\pi}{\gamma}, \lab{\sigma}{\beta\rightarrow\gamma} \seq \ldots, \lab{x}{\gamma}$}
\LeftLabel{$\il$}
\UnaryInfC{$\ldots, \relat{x}{\pi\cup\sigma}, \lab{x}{\beta}, \lab{x\cap\pi}{\gamma}, \lab{\sigma}{\beta\rightarrow\gamma} \seq \ldots, \lab{x}{\gamma}$}
\LeftLabel{$\Lrule{\rightarrow}$}
\BinaryInfC{$\relat{x\cap\pi}{\pi}, \relat{x}{\pi\cup\sigma}, \lab{x}{\beta}, \lab{\pi}{\beta\rightarrow\gamma}, \lab{\sigma}{\beta\rightarrow\gamma}, \Gamma\seq\Delta, \lab{x}{\gamma}$}
\LeftLabel{$\il$}
\UnaryInfC{$\relat{x}{\pi\cup\sigma}, \lab{x}{\beta}, \lab{\pi}{\beta\rightarrow\gamma}, \lab{\sigma}{\beta\rightarrow\gamma}, \Gamma\seq\Delta, \lab{x}{\gamma}$}
\LeftLabel{$\Rrule{\rightarrow}$}
\UnaryInfC{$\lab{\pi}{\beta\rightarrow\gamma}, \lab{\sigma}{\beta\rightarrow\gamma}, \Gamma\seq\Delta, \lab{\pi\cup\sigma}{\beta\rightarrow\gamma}$}
\DisplayProof 
\end{center}
where the subderivation $\mathcal{D}$ is of the form 
\begin{center}
\myfontsize 
\AxiomC{By~ind.~hyp.}
\dashedLine
\UnaryInfC{$\ldots, \lab{x\cap\pi}{\gamma}, \lab{x\cap\sigma}{\gamma} \seq \ldots, \lab{(x\cap\pi)\cup(x\cap\sigma)}{\gamma}$}
\AxiomC{Lemma~\ref{lem:gen:ax}~(\ref{gen:ax:i})}
\dashedLine
\UnaryInfC{$\relat{x}{(x\cap\pi)\cup(x\cap\sigma)}, \lab{(x\cap\pi)\cup(x\cap\sigma)}{\gamma} \seq \lab{x}{\gamma}$}
\LeftLabel{$\cut$}
\BinaryInfC{$\ldots, \relat{x}{(x\cap\pi)\cup(x\cap\sigma)}, \lab{x\cap\pi}{\gamma}, \lab{x\cap\sigma}{\gamma} \seq \ldots, \lab{x}{\gamma}$}
\LeftLabel{$\dis$}
\UnaryInfC{$\ldots, \relat{x}{\pi\cup\sigma}, \lab{x\cap\pi}{\gamma}, \lab{x\cap\sigma}{\gamma} \seq \ldots, \lab{x}{\gamma}$}
\DisplayProof
\end{center}
The case $\alpha=\bot$ can be treated in the same way as case 1. The case $\alpha = \beta\wedge\gamma$ is straightforward.    
\end{proof}

We are now ready to prove the desired admissibility results.

\begin{lemma}\label{lem:mp:flat:admissible}
Let $\Logic$ be any of the logics from Table~\ref{tab:team:logics} and let $\Gsystem{\Logic}$ be the corresponding labelled sequent calculus. The modus ponens rule $\MP$ and the flatness rule $\FL$ are admissible in $\Gsystem{\Logic}$.
\end{lemma}

\begin{proof}
To prove the admissibility of $\MP$, suppose that $\Gamma\seq \lab{x}{\varphi}$ and $\Delta\seq \lab{x}{\varphi\rightarrow\psi}$ are both derivable in $\Gsystem{\Logic}$ by proof trees $\mathcal{D}_1$ and $\mathcal{D}_2$, respectively. By the invertibility of $\Rrule{\rightarrow}$, the admissibility of substitution and an application of $\rf$, we may transform $\mathcal{D}_2$ into a derivation $\mathcal{D}_2'$ for $\lab{x}{\varphi},\Delta \seq \lab{x}{\psi}$. But then, by applying the cut rule to $\mathcal{D}_1$ and $\mathcal{D}_2'$, we also obtain a derivation for $\Gamma,\Delta \seq \lab{x}{\psi}$, as desired.

Next, we prove the admissibility of the rule $\FL$. If $\Logic\in\{\InqB,\PID\}$, the proof proceeds by an easy induction on the structure of $\alpha$, which is $\vee$-free. Proofs can be found in \cite[Prop.~3.2.14]{muller:msc:2023} and \cite[Prop.~4.6]{muller:2026}. Suppose now that $\Logic\in\{\InqBtensor,\PIDtensor\}$. Furthermore, consider an arbitrary infinite sequence $u_1,u_2,u_3,\ldots$ of variables from $\Svar$. For every $k\geq 1$, let $\unionlab{k}$ be the label defined by $\unionlab{k} := u_1\cup\ldots\cup u_k$, where the union operator is assumed to be \emph{left-associative}. So, for example, $\unionlab{4}$ stands for $((u_{1}\cup u_{2})\cup u_{3})\cup u_{4}$. Using induction on $k\geq 1$, we first prove the admissibility of the following restricted rule:
\begin{center}
\myfontsize 
\AxiomC{$\relat{v}{\unionlab{k}},\Gamma\seq\Delta,\lab{v}{\alpha}$}
\LeftLabel{$\FLfin{k}$}
\UnaryInfC{$\Gamma\seq\Delta,\lab{\unionlab{k}}{\alpha}$}
\DisplayProof 
\end{center}
where $\alpha$ must be a classical formula and $v\in\Svar$ must be fresh. For the base case, let $k=1$ and suppose that $\relat{v}{u_1},\Gamma\seq\Delta,\lab{v}{\alpha}$ is derivable in $\Gsystem{\Logic}$. Using the admissibility of substitution, this yields a derivation for $\relat{u_1}{u_1},\Gamma\seq\Delta,\lab{u_1}{\alpha}$. By an application of $\rf$, it now follows that $\Gamma\seq\Delta,\lab{u_1}{\alpha}$ is derivable. 

For the inductive step, let $k>1$ and suppose that the premise $\relat{v}{\unionlab{k}},\Gamma\seq\Delta,\lab{v}{\alpha}$ of the rule $\FLfin{k}$ is derivable in the system $\Gsystem{\Logic}$ by a proof tree $\mathcal{D}$. Using the admissibility of substitution, weakening, contraction and cut, we may then derive the conclusion of $\FLfin{k}$ in the following way:
\begin{center}
\myfontsize
\def\defaultHypSeparation{\hskip .1in}
\AxiomC{$\mathcal{D}$}
\noLine
\UnaryInfC{$\relat{v}{\unionlab{k}},\Gamma\seq\Delta,\lab{v}{\alpha}$}
\LeftLabel{$\WL$}
\UnaryInfC{$\relat{v}{\unionlab{k}}, \relat{v}{u_k}, \relat{u_k}{\unionlab{k}}, \Gamma\seq\Delta,\lab{v}{\alpha}$}
\LeftLabel{$\tr$}
\UnaryInfC{$\relat{v}{u_k}, \relat{u_k}{\unionlab{k}}, \Gamma\seq\Delta,\lab{v}{\alpha}$}
\LeftLabel{$\ur$}
\UnaryInfC{$\relat{v}{u_k}, \Gamma\seq\Delta,\lab{v}{\alpha}$\vphantom{$()$}}
\LeftLabel{$(u_{k}/v)$}
\UnaryInfC{$\relat{u_k}{u_k}, \Gamma\seq\Delta,\lab{u_k}{\alpha}$\vphantom{$()$}}
\LeftLabel{$\rf$}
\UnaryInfC{$\Gamma\seq\Delta,\lab{u_k}{\alpha}$\vphantom{$()$}}
\AxiomC{$\mathcal{D}$}
\noLine
\UnaryInfC{$\relat{v}{\unionlab{k}},\Gamma\seq\Delta,\lab{v}{\alpha}$}
\LeftLabel{$\WL$}
\UnaryInfC{$\relat{v}{\unionlab{k}}, \relat{v}{\unionlab{k-1}}, \relat{\unionlab{k-1}}{\unionlab{k}}, \Gamma\seq\Delta,\lab{v}{\alpha}$}
\LeftLabel{$\tr$}
\UnaryInfC{$\relat{v}{\unionlab{k-1}}, \relat{\unionlab{k-1}}{\unionlab{k}}, \Gamma\seq\Delta,\lab{v}{\alpha}$}
\LeftLabel{$\ur$}
\UnaryInfC{$\relat{v}{\unionlab{k-1}}, \Gamma\seq\Delta,\lab{v}{\alpha}$}
\LeftLabel{IH}
\UnaryInfC{$\Gamma\seq\Delta,\lab{\unionlab{k-1}}{\alpha}$}
\AxiomC{Lemma~\ref{lem:union:closure}}
\dashedLine
\UnaryInfC{$\lab{\unionlab{k-1}}{\alpha},\lab{u_k}{\alpha} \seq \lab{\unionlab{k}}{\alpha}$}
\LeftLabel{$\cut$}
\BinaryInfC{$\lab{u_k}{\alpha},\Gamma\seq\Delta,\lab{\unionlab{k}}{\alpha}$}
\LeftLabel{$\cut$}
\BinaryInfC{$\Gamma,\Gamma \seq \Delta,\Delta,\lab{\unionlab{k}}{\alpha}$}
\LeftLabel{$\CL,\CR$}
\UnaryInfC{$\Gamma\seq\Delta,\lab{\unionlab{k}}{\alpha}$}
\def\defaultHypSeparation{\hskip .2in}
\DisplayProof
\end{center}
where ``IH'' denotes an application of the induction hypothesis. This concludes the induction. In order to prove the admissibility of the unrestricted rule $\FL$, suppose now that the sequent $\relat{v}{\pi}, \Gamma\seq\Delta, \lab{v}{\alpha}$ is derivable in $\Gsystem{\Logic}$ by a proof tree $\mathcal{D}$, where $\pi\in\LABint$ is an arbitrary label and $v\in\Svar$ is a fresh variable not occurring in $\Gamma\seq\Delta, \lab{\pi}{\alpha}$. Moreover, let $n:= 2^{|\Prop|}$ and let $u_{1},\ldots, u_{n}$ be a collection of $n$ pairwise distinct variables with $v\notin\{u_{1},\ldots, u_{n}\}$ such that $u_{1},\ldots, u_{n}$ also do not occur in $\Gamma\seq\Delta, \lab{\pi}{\alpha}$. Finally, let $\altunlab{n}$ be the label given by $\altunlab{n} := u_{1}\cup\ldots\cup u_{n}$, where the union operator is again assumed to be left-associative. Using the admissibility of $\FLfin{n}$, we may then derive the conclusion of $\FL$ as follows:
\begin{center}
\myfontsize
\AxiomC{Lemma~\ref{lem:gen:ax}~(\ref{gen:ax:ii})}
\dashedLine
\UnaryInfC{$\equat{\pi}{\emptyset}, \Gamma\seq\Delta, \lab{\pi}{\alpha}$}
\AxiomC{$\mathcal{D}$}
\noLine
\UnaryInfC{$\relat{v}{\pi}, \Gamma\seq\Delta, \lab{v}{\alpha}$}
\LeftLabel{$\WL$}
\UnaryInfC{$\relat{v}{\pi}, \relat{v}{\altunlab{n}}, \equat{\pi}{\altunlab{n}}, \Gamma\seq\Delta, \lab{v}{\alpha}$}
\LeftLabel{$\tr$}
\UnaryInfC{$\relat{v}{\altunlab{n}}, \equat{\pi}{\altunlab{n}}, \Gamma\seq\Delta, \lab{v}{\alpha}$}
\LeftLabel{$\FLfin{n}$}
\UnaryInfC{$\equat{\pi}{\altunlab{n}}, \Gamma\seq\Delta, \lab{\altunlab{n}}{\alpha}$}
\AxiomC{Lemma~\ref{lem:gen:ax}~(\ref{gen:ax:i})}
\dashedLine
\UnaryInfC{$\equat{\pi}{\altunlab{n}}, \lab{\altunlab{n}}{\alpha} \seq \lab{\pi}{\alpha}$}
\LeftLabel{$\cut$}
\BinaryInfC{$\equat{\pi}{\altunlab{n}}, \equat{\pi}{\altunlab{n}}, \Gamma\seq\Delta, \lab{\pi}{\alpha}$}
\LeftLabel{$\CL$}
\UnaryInfC{$\equat{\pi}{\altunlab{n}}, \Gamma\seq\Delta, \lab{\pi}{\alpha}$}
\LeftLabel{$\fin$}
\BinaryInfC{$\Gamma\seq\Delta,\lab{\pi}{\alpha}$}
\DisplayProof
\end{center}
where the notation ``$\equat{\sigma}{\tau}$'' is used as an abbreviation for the pair of relational atoms ``$\relat{\sigma}{\tau}, \relat{\tau}{\sigma}$''.
\end{proof}

In the case $\Logic\in\{\InqBtensor,\PIDtensor\}$, our admissibility proof for $\FL$ crucially uses the rule $\fin$ and thus relies on the assumption that $\Prop$ is finite. Using the admissibility of the rules $\MP$ and $\FL$, we now prove that each of our calculi is complete with respect to the corresponding Hilbert-style system from Figure~\ref{fig:hilbert:systems}.

\begin{lemma}\label{lem:complete:wrt:hilbert}
Let $\Logic$ be any of the logics from Table~\ref{tab:team:logics}, let $\Hsystem{\Logic}$ be the Hilbert-style system and let $\Gsystem{\Logic}$ be the labelled sequent calculus for $\Logic$. If $\Gamma\provH\varphi$ holds in $\Hsystem{\Logic}$, then $\Gamma\provG\varphi$ holds in $\Gsystem{\Logic}$. 
\end{lemma}

\begin{proof}
By induction on the length of a proof for $\Gamma\provH\varphi$ in the Hilbert system $\Hsystem{\Logic}$. The inductive step is easy, since we already know that modus ponens is admissible in $\Gsystem{\Logic}$ (Lemma~\ref{lem:mp:flat:admissible}). Hence, it suffices to show that all axioms of $\Hsystem{\Logic}$ are provable in $\Gsystem{\Logic}$. A few cases are considered in Appendix~\ref{app:proof:complete:wrt:hilbert}. 
\end{proof}

\begin{theorem}[Soundness and Completeness]
Let $\Logic$ be any of the logics from Table~\ref{tab:team:logics}. The labelled sequent calculus $\Gsystem{\Logic}$ is sound and complete with respect to $\Logic$. That is, for every set of formulas $\Gamma\cup\{\varphi\}\subseteq\Langlog{\Logic}$, it is the case that $\Gamma\provG\varphi$ holds in $\Gsystem{\Logic}$ if and only if $\Gamma\semL\varphi$ holds in $\Logic$. 
\end{theorem}

\begin{proof}
For the soundness direction, one may first use induction on the structure of a derivation to show that every sequent derivable in $\Gsystem{\Logic}$ is valid. Soundness then follows as an immediate consequence (cf.\ \cite[Prop.~3.3.4]{muller:msc:2023}). The completeness part follows directly from Theorem~\ref{th:completeness:hilbert} and Lemma~\ref{lem:complete:wrt:hilbert}.   
\end{proof}

\section{Proof Search}\label{sec:proof:search}

We now want to describe a terminating proof search algorithm for the logics presented in Section~\ref{sec:preliminaries}. To avoid certain complications arising from the complex syntax of the labels presented in Definition~\ref{def:labels}, we first introduce an alternative class of labelled sequent calculi for our logics, in which labels are defined as \emph{sets of variables}.\footnote{A very similar calculus for $\InqB$ was already presented by Litak and Sano \cite[Remark~4]{litak:sano:2026}. In fact, their system essentially coincides with the calculus for $\InqB$ introduced in Definition~\ref{def:terminating:calculi} below. The only noteworthy differences are that Litak and Sano use slightly different axioms and require labels to be non-empty.} This eliminates the need for relational atoms and order rules, which simplifies the proof of termination significantly. Throughout this section, we will assume an infinite set $\Vvar$ of variables ranging over propositional \emph{valuations}. Elements of $\Vvar$ are denoted by the meta-variables $u$ and $v$.  

\begin{definition}[Labels]
A \emph{label} is a finite set of variables $\pi\subseteq\Vvar$.   
\end{definition}
 
\looseness=-1 A label of the form $\pi = \{v_{1},\ldots,v_{k}\}$ stands for the finite team consisting of the valuations represented by $v_{1},\ldots,v_{k}$ (which do not need to be distinct). Arbitrary labels are denoted by $\pi$, $\sigma$, $\tau$. A \emph{split} of a label $\pi$ is a pair of labels $(\sigma,\tau)$ such that $\pi = \sigma\cup\tau$. The set of all splits of $\pi$ is denoted by $\Splits{\pi}$, so $\Splits{\pi} := \{(\sigma,\tau) \mid \pi = \sigma\cup\tau\}$. A \emph{labelled formula} is an expression of the form $\lab{\pi}{\varphi}$, where $\pi\subseteq\Vvar$ is a label and $\varphi\in\LangTL$ is a formula. If $\pi = \{v\}$ is a singleton label, then we write $\lab{v}{\varphi}$ instead of $\lab{\pi}{\varphi}$. As before, $\lab{\pi}{\varphi}$ may be read as ``$\varphi$ is satisfied by the team $\pi$''. A \emph{sequent} is an expression of the form $\Gamma\seq\Delta$, where $\Gamma$ and $\Delta$ are finite multisets of labelled formulas 
(thus, sequents do \emph{not} contain relational atoms).

\begin{definition}[Interpretation of Labels]
An \emph{interpretation} is a function $\I :\Vvar\rightarrow \Val$ that assigns, to each variable $v\in\Vvar$, a valuation $\I(v)\in\Val$. Such a function $\I$ is extended to labels by putting $\I(\emptyset) := \emptyset$ as well as $\I(\pi) := \{\I(v_{1}),\ldots,\I(v_{k})\}$ for every label $\pi=\{v_{1},\ldots,v_{k}\}$ with $k\geq 1$. 
\end{definition}

\begin{figure}[t] 
\centering\myfontsize 
\begin{tabularx}{\textwidth}{|@{\hskip .5cm}Y@{\hskip .5cm}|}
\hline
\rule{0pt}{.7cm} 
\AxiomC{\vphantom{$\Gamma,\Delta$}}
\LeftLabel{$\ax$}
\UnaryInfC{$\lab{v}{p}, \extseq, \lab{v}{p}$}
\DisplayProof
\hspace{2em}
\AxiomC{\vphantom{$\Gamma,\Delta$}}
\LeftLabel{$\axBot$}
\UnaryInfC{$\lab{v}{\bot}, \extseq$}
\DisplayProof
\hspace{2em}
\AxiomC{$\lab{v}{p}, \lab{\A}{p}, \extseq$}
\LeftLabel{$\Latom$}
\RightLabel{$(v\in\A)$}
\UnaryInfC{$\lab{\A}{p}, \extseq$}
\DisplayProof
\hspace{2em}
\AxiomC{$[\extseq, \lab{v}{p}]_{v\in\A}$}
\LeftLabel{$\Ratom$}
\RightLabel{$(|\A|\neq 1)$}
\UnaryInfC{$\extseq, \lab{\A}{p}$}
\DisplayProof
\\[\mylength]
\AxiomC{$\lab{v}{\bot}, \lab{\A}{\bot}, \extseq$}
\LeftLabel{$\Lrule{\bot}$}
\RightLabel{$(v\in\A)$}
\UnaryInfC{$\lab{\A}{\bot}, \extseq$}
\DisplayProof
\hspace{1em}
\AxiomC{$[\extseq, \lab{v}{\bot}]_{v\in\A}$}
\LeftLabel{$\Rrule{\bot}$}
\RightLabel{$(|\A|\neq 1)$}
\UnaryInfC{$\extseq, \lab{\A}{\bot}$}
\DisplayProof
\hspace{1em}
\AxiomC{$\lab{\A}{\varphi},\lab{\A}{\psi}, \extseq$}
\LeftLabel{$\Lrule{\wedge}$}
\UnaryInfC{$\lab{\A}{\varphi\wedge\psi}, \extseq$}
\DisplayProof
\hspace{1em}
\AxiomC{$\extseq, \lab{\A}{\varphi}$}
\AxiomC{$\extseq, \lab{\A}{\psi}$}
\LeftLabel{$\Rrule{\wedge}$}
\BinaryInfC{$\extseq, \lab{\A}{\varphi\wedge\psi}$}
\DisplayProof
\\[\mylength]
\AxiomC{$\lab{\A}{\varphi}, \extseq$}
\AxiomC{$\lab{\A}{\psi}, \extseq$}
\LeftLabel{$\Lrule{\globdis}$}
\BinaryInfC{$\lab{\A}{\varphi\globdis\psi}, \extseq$}
\DisplayProof
\hspace{.5em}
\AxiomC{$\extseq, \lab{\A}{\varphi}, \lab{\A}{\psi}$}
\LeftLabel{$\Rrule{\globdis}$}
\UnaryInfC{$\extseq, \lab{\A}{\varphi\globdis\psi}$}
\DisplayProof
\hspace{.5em}
\AxiomC{$\lab{\A}{\varphi\rightarrow\psi}, \extseq, \lab{\B}{\varphi}$}
\AxiomC{$\lab{\A}{\varphi\rightarrow\psi}, \lab{\B}{\psi}, \extseq$}
\LeftLabel{$\Lrule{\rightarrow}$}
\RightLabel{$(\B\subseteq\A)$}
\BinaryInfC{$\lab{\A}{\varphi\rightarrow\psi}, \extseq$}
\DisplayProof
\\[\mylength]
\AxiomC{$[\lab{\B}{\varphi}, \extseq, \lab{\B}{\psi}]_{\B\subseteq\A}$}
\LeftLabel{$\Rrule{\rightarrow}$}
\UnaryInfC{$\extseq, \lab{\A}{\varphi\rightarrow\psi}$}
\DisplayProof
\hspace{0pt}
\AxiomC{$[\lab{\B}{\varphi},\lab{\C}{\psi},\extseq]_{(\B,\C)\in\Splits{\A}}$}
\LeftLabel{$\Lrule{\tensor}$}
\UnaryInfC{$\lab{\A}{\varphi\tensor\psi},\extseq$}
\DisplayProof
\hspace{0pt}
\def\defaultHypSeparation{\hskip .1in}
\AxiomC{$\extseq,\lab{\A}{\varphi\tensor\psi},\lab{\B}{\varphi}$}
\AxiomC{$\extseq,\lab{\A}{\varphi\tensor\psi},\lab{\C}{\psi}$}
\LeftLabel{$\Rrule{\tensor}$}
\RightLabel{$(\A = \B\cup\C)$}
\BinaryInfC{$\extseq,\lab{\A}{\varphi\tensor\psi}$}
\DisplayProof
\def\defaultHypSeparation{\hskip .2in}
\\[\mylength]
\AxiomC{$\lab{u}{p}, \lab{v}{p}, \lab{\A}{\dep{p}}, \extseq$}
\AxiomC{$\lab{\A}{\dep{p}}, \extseq, \lab{u}{p}, \lab{v}{p}$}
\LeftLabel{$\Lrule{=}$}
\RightLabel{$(u,v\in\A)$}
\BinaryInfC{$\lab{\A}{\dep{p}}, \extseq$}
\DisplayProof
\hspace{1em}
\AxiomC{$[\lab{u}{p}, \extseq, \lab{v}{p}]_{(u,v)\in\A\times\A}$}
\LeftLabel{$\Rrule{=}$}
\UnaryInfC{$\extseq, \lab{\A}{\dep{p}}$}
\DisplayProof
\\[\mylength] 
\hline
\end{tabularx}
\caption{Axioms and rules of our terminating labelled sequent calculi.}\label{fig:rules:terminating:calculi}
\end{figure} 

\looseness=-1 A labelled formula $\lab{\pi}{\varphi}$ is \emph{satisfied} by an interpretation $\I$, if $\I(\pi)\sem\varphi$. And a sequent $\Gamma\seq\Delta$ is \emph{valid}, if for each interpretation $\I$, we have: if $\I$ satisfies all elements of $\Gamma$, then $\I$ satisfies at least one element of $\Delta$. The axioms and rules of our terminating calculi are presented in Figure~\ref{fig:rules:terminating:calculi}. Note that, in the figure, we use square brackets to denote collections of premises. So, for example, an application of the rule $\Ratom$ with principal formula $\lab{\pi}{p}$ now has exactly $|\pi|$ premises---namely, one premise for each variable $v\in\pi$. In the special case $\pi = \emptyset$, we treat $\Ratom$, $\Rrule{\bot}$ and $\Rrule{=}$ as \emph{initial sequents} (i.e., rules with zero premises). Some of our rules are also subject to additional side conditions, which are indicated on the right-hand side of the ``inference line''. For instance, the pair $(\sigma,\tau)$ must be a split of $\pi$ in applications of $\Rrule{\tensor}$. The rules $\Latom$, $\Lrule{\bot}$, $\Lrule{\rightarrow}$, $\Rrule{\tensor}$, $\Lrule{=}$ are called \emph{cumulative}, since their principal formulas are repeated in their premises. 

\begin{definition}\label{def:terminating:calculi}
Let $\Logic$ be any of the logics from Table~\ref{tab:team:logics} and let $\Conn$ be the corresponding set of connectives (see the second column of Table~\ref{tab:team:logics}). We write $\GTsystem{\Logic}$ for the labelled sequent calculus consisting of the initial sequents ($\ax$ and $\axBot$), the rules for atoms ($\Latom$ and $\Ratom$) and the left and right rules for the connectives in $\Conn$ (see Figure~\ref{fig:rules:terminating:calculi}).\footnote{So, for example, the calculus $\GTsystem{\InqB}$ consists of the initial sequents, the rules for atoms and the rules for $\bot,\wedge,\globdis,\rightarrow$. The calculus $\GTsystem{\PIDtensor}$, on the other hand, consists of the initial sequents, the rules for atoms and the rules for $\bot,\wedge,\globdis,\rightarrow,\dep{\cdot},\tensor$.} For any set of formulas $\Gamma\cup\{\varphi\}\subseteq\Langlog{\Logic}$, we also write $\Gamma\provGT\varphi$ and say that $\varphi$ is \emph{provable} from $\Gamma$ in $\GTsystem{\Logic}$, if there exists a finite subset $\Delta\subseteq\Gamma$ and a label $\pi\subseteq\Vvar$ with $|\pi| = 2^{|\Prop|}$ such that the sequent $\lab{\pi}{\Delta}\seq\lab{\pi}{\varphi}$ is derivable in $\GTsystem{\Logic}$, where $(\lab{\pi}{\Delta}) := \{\lab{\pi}{\psi} \mid \psi\in\Delta\}$.    
\end{definition}

\begin{lemma}[Soundness]\label{lem:soundness:terminating:system}
Let $\Logic$ be one of the logics from Table~\ref{tab:team:logics}. The calculus $\GTsystem{\Logic}$ is sound with respect to $\Logic$. That is, for every $\Gamma\cup\{\varphi\}\in\Langlog{\Logic}$, if $\Gamma\provGT\varphi$ holds in $\GTsystem{\Logic}$, then $\Gamma\semL\varphi$ holds in $\Logic$.   
\end{lemma}

\begin{proof}
Using induction on the structure of a derivation, one can first show that every sequent derivable in $\GTsystem{\Logic}$ is valid. Suppose now that $\Gamma\provGT\varphi$ holds in $\GTsystem{\Logic}$, i.e., there exists a finite subset $\Delta\subseteq\Gamma$ and a label $\pi=\{v_{1},\ldots,v_{n}\}$ with $n = 2^{|\Prop|}$ such that $\lab{\pi}{\Delta}\seq\lab{\pi}{\varphi}$ is derivable in $\GTsystem{\Logic}$. By what was said above, $\lab{\pi}{\Delta}\seq\lab{\pi}{\varphi}$ must then be valid. Let $T\subseteq\Val$ be an arbitrary team and assume that $T\sem\psi$ for all $\psi\in\Gamma$. If $T=\emptyset$, then we trivially have $T\sem\varphi$. Thus, let us suppose that $T\neq\emptyset$. Since $|\Val| = 2^{|\Prop|}$, we know that $T$ is of the form $T=\{u_{1},\ldots,u_{k}\}$ for some $k$ with $1\leq k\leq 2^{|\Prop|}$. Let now $\I$ be the interpretation given by $\I(v_{i}) := u_{i}$ for all $i$ with $1\leq i\leq k$, and $\I(v_{i}) := u_{1}$ for all $i$ with $k < i \leq 2^{|\Prop|}$. Then, clearly, we have $\I(\pi) = T$. Hence, from $T\sem\Gamma$ and $\Delta\subseteq\Gamma$, we obtain $\I(\pi)\sem\Delta$. By the validity of $\lab{\pi}{\Delta}\seq\lab{\pi}{\varphi}$, this yields $\I(\pi)\sem\varphi$ and thus $T\sem\varphi$. Since $T$ was an arbitrary team with $T\sem\Gamma$, this shows that $\Gamma\semL\varphi$.    
\end{proof}

Throughout the rest of this section, let $\Logic$ be any of the logics from Table~\ref{tab:team:logics}. We will present an algorithm that performs root-first proof search in the calculus $\GTsystem{\Logic}$. Our algorithm will construct a proof search tree for a given sequent by applying the rules of $\GTsystem{\Logic}$ bottom-up to the topmost sequents in the tree. The rules will be applied exhaustively, as all possibilities need to be explored, but non-redundantly, to avoid repetitions. As a consequence, the search terminates within a finite number of steps, yielding either a proof or a \emph{saturated branch} (Lemma~\ref{lemma:termination}). Intuitively, a branch is saturated if its topmost sequent is not an axiom and no further rules can be non-redundantly applied to it. We will then show that a saturated branch can always be used to define an interpretation that falsifies every sequent in the branch (Lemma~\ref{lem:properties:interpretation}). This will allow us to establish the completeness of our calculi. We start by introducing some terminology.

\begin{definition}[Proof Search Tree]
A \emph{proof search tree} is a finite tree of sequents $\Tree$ such that (1) the root node of $\Tree$ is of the form $\lab{\pi}{\Gamma}\seq\lab{\pi}{\varphi}$ for some label $\pi$ with $|\pi| = 2^{|\Prop|}$, and (2) $\Tree$ is built up from root-first applications of the rules of $\GTsystem{\Logic}$. A \emph{branch} in a proof search tree $\Tree$ is a finite sequence $B$ of consecutive sequents in $\Tree$ such that the first sequent in $B$ is the root node of $\Tree$ and the last sequent in $B$ is one of the leaf nodes of $\Tree$. The last sequent in a branch $B$ is also referred to as the \emph{topmost sequent} of $B$.       
\end{definition}

A branch $B$ is \emph{closed} if the topmost sequent of $B$ is an instance of an axiom (note that sequents of the form $\Gamma\seq\Delta,\lab{\emptyset}{p}$ or $\Gamma\seq\Delta,\lab{\emptyset}{\bot}$ or $\Gamma\seq\Delta,\lab{\emptyset}{\dep{p}}$ now also count as axioms). A branch that is not closed is \emph{open}. A proof search tree $\Tree$ is closed if every branch in $\Tree$ is closed. Given a branch $B$ of the form $\Gamma_{1}\seq\Delta_{1}, \Gamma_{2}\seq\Delta_{2},\ldots,\Gamma_{n}\seq\Delta_{n}$, we also define $\Gdown := \bigcup_{1\leq i\leq n} \Gamma_{i}$ and $\Ddown := \bigcup_{1\leq i\leq n} \Delta_{i}$, so $\Gdown$ is the union of all the \emph{antecedents} and $\Ddown$ is the union of all the \emph{succedents} of the sequents in $B$. Moreover, we write $\Ltop := \Gamma_{n}$ for the antecedent and $\Rtop := \Delta_{n}$ for the succedent of the \emph{topmost sequent} of $B$.

\begin{definition}[Saturated Branch]
\label{def:saturation}
A branch $B$ is \emph{saturated}, if $B$ is open and each of the following holds:
\begin{enumerate}
\setlength{\itemsep}{0pt}
\setlength{\parskip}{0pt}
\item\label{def:sat:i} for $\alpha\in\Prop\cup\{\bot\}$: if $(\lab{\pi}{\alpha})\in\Gdown$, then also $(\lab{v}{\alpha})\in\Gdown$ for all $v\in\pi$; 
\item\label{def:sat:ii} for $\alpha\in\Prop\cup\{\bot\}$: if $(\lab{\pi}{\alpha})\in\Ddown$ and $\pi\neq\emptyset$, then $(\lab{v}{\alpha})\in\Ddown$ for some $v\in\pi$;

\item\label{def:sat:v} if $(\lab{\pi}{\varphi\wedge\psi})\in\Gdown$, then both $(\lab{\pi}{\varphi})\in\Gdown$ and $(\lab{\pi}{\psi})\in\Gdown$; 
\item\label{def:sat:vi} if $(\lab{\pi}{\varphi\wedge\psi})\in\Ddown$, then $(\lab{\pi}{\varphi})\in\Ddown$ or $(\lab{\pi}{\psi})\in\Ddown$; 

\item\label{def:sat:vii} if $(\lab{\pi}{\varphi\globdis\psi})\in\Gdown$, then $(\lab{\pi}{\varphi})\in\Gdown$ or $(\lab{\pi}{\psi})\in\Gdown$;
\item\label{def:sat:viii} if $(\lab{\pi}{\varphi\globdis\psi})\in\Ddown$, then both $(\lab{\pi}{\varphi})\in\Ddown$ and $(\lab{\pi}{\psi})\in\Ddown$;

\item\label{def:sat:ix} if $(\lab{\pi}{\varphi\rightarrow\psi})\in\Gdown$ and $\sigma\subseteq\pi$, then $(\lab{\sigma}{\varphi})\in\Ddown$ or $(\lab{\sigma}{\psi})\in\Gdown$;
\item\label{def:sat:x} if $(\lab{\pi}{\varphi\rightarrow\psi})\in\Ddown$, then $(\lab{\sigma}{\varphi})\in\Gdown$ and $(\lab{\sigma}{\psi})\in\Ddown$ for some $\sigma\subseteq\pi$;

\item\label{def:sat:xi} if $(\lab{\pi}{\varphi\tensor\psi})\in\Gdown$, then $(\lab{\sigma}{\varphi})\in\Gdown$ and $(\lab{\tau}{\psi})\in\Gdown$ for some split $(\sigma,\tau)\in\Splits{\pi}$;
\item\label{def:sat:xii} if $(\lab{\pi}{\varphi\tensor\psi})\in\Ddown$ and $(\sigma,\tau)\in\Splits{\pi}$, then $(\lab{\sigma}{\varphi})\in\Ddown$ or $(\lab{\tau}{\psi})\in\Ddown$;

\item\label{def:sat:xiii} if $(\lab{\pi}{\dep{p}})\in\Gdown$ and $u,v\in\pi$, then [$(\lab{u}{p})\in\Gdown$, $(\lab{v}{p})\in\Gdown$] or [$(\lab{u}{p})\in\Ddown$, $(\lab{v}{p})\in\Ddown$];
\item\label{def:sat:xiv} if $(\lab{\pi}{\dep{p}})\in\Ddown$ and $\pi\neq\emptyset$, then $(\lab{u}{p})\in\Gdown$ and $(\lab{v}{p})\in\Ddown$ for some $(u,v)\in\pi\times\pi$.
\end{enumerate}
\end{definition}

\begin{figure}
\centering
\resizebox{\textwidth}{!}{
\begin{minipage}{1.15\linewidth}
\begin{algorithm}[H]  
\caption{An excerpt from the procedure $\ProofSearch$}
\label{alg:proof:search}
\KwIn{A finite set of formulas $\Gamma\subseteq\Langlog{\Logic}$ and a formula $\varphi\in\Langlog{\Logic}$.}
\KwOut{If $\Gamma\provGT\varphi$ holds in $\GTsystem{\Logic}$, then the procedure outputs ``success''. Otherwise, the procedure outputs ``failure''.}
\medskip
Select a label $\pi$ with $|\pi| = 2^{|\Prop|}$ and initialize $\Tree$ as the proof search tree consisting only of the root node $\lab{\pi}{\Gamma}\seq\lab{\pi}{\varphi}$\; \label{alg:init}
\While{true}{ 
\Comment{root-first applications of $\Rrule{\rightarrow}$} 
\label{alg:start:expansion}
\While{$\Tree$ has an open branch $B$ with a topmost sequent of the form $\Pi\seq\Delta,\lab{\sigma}{\psi\rightarrow\chi}$
\label{alg:impl:right:start}
}{
extend $\Tree$ by writing, for every subset $\tau\subseteq\sigma$, the sequent $\lab{\tau}{\psi},\Pi\seq\Delta,\lab{\tau}{\chi}$ on top of $B$\; 
} \label{alg:impl:right:end}

\Comment{root-first applications of $\Lrule{\rightarrow}$}
\While{$\Tree$ has an open branch $B$ with $\lab{\sigma}{\psi\rightarrow\chi}\in\Ltop$ and there exists some $\tau\subseteq\sigma$ such that $(\lab{\tau}{\psi})\notin\Ddown$ and $(\lab{\tau}{\chi})\notin\Gdown$
\label{alg:impl:left:start}
}{
extend $\Tree$ by writing both $\Ltop\seq\Rtop,\lab{\tau}{\psi}$ and $\lab{\tau}{\chi},\Ltop\seq\Rtop$ on top of $B$\; 
} \label{alg:impl:left:end}

\Comment{root-first applications of $\Rrule{\tensor}$}
\While{$\Tree$ has an open branch $B$ with $(\lab{\sigma}{\psi\tensor\chi})\in\Rtop$ and there exists some $(\tau_{1},\tau_{2})\in\Splits{\sigma}$ such that $(\lab{\tau_{1}}{\psi})\notin\Ddown$ and $(\lab{\tau_{2}}{\chi})\notin\Ddown$
\label{alg:tensor:right:start}
}{
extend $\Tree$ by writing both $\Ltop\seq\Rtop,\lab{\tau_{1}}{\psi}$ and $\Ltop\seq\Rtop,\lab{\tau_{2}}{\chi}$ on top of $B$\; 
} \label{alg:tensor:right:end}

\Comment{root-first applications of $\Lrule{=}$}
\While{$\Tree$ has an open branch $B$ with $(\lab{\sigma}{\dep{p}})\in\Ltop$ and $\exists u,v\in\sigma$ s.t.\ $[(\lab{u}{p})\notin\Gdown$ or $(\lab{v}{p})\notin\Gdown]$ and $[(\lab{u}{p})\notin\Ddown$ or $(\lab{v}{p})\notin\Ddown]$
\label{alg:constancy:left:start}
}{
extend $\Tree$ by writing both $\lab{u}{p},\lab{v}{p},\Ltop\seq\Rtop$ and $\Ltop\seq\Rtop,\lab{u}{p},\lab{v}{p}$ on top of $B$\; 
} \label{alg:constancy:left:end}

\Comment{$\ldots$ add here similar while-loops for the other rules $\ldots$} 

\lIf(\tcp*[f]{$\Tree$ is now a derivation for the root sequent}){$\Tree$ is closed}{\Return{} ``success''} \label{alg:end:condition:success}
\lIf(\tcp*[f]{$B$ determines a counter-interpretation for the root sequent}){$\Tree$ has a saturated branch $B$}{\Return{} ``failure''} \label{alg:end:condition:failure}
}
\end{algorithm}
\end{minipage}
}
\end{figure}

We are now ready to describe the desired proof search procedure for the calculus $\GTsystem{\Logic}$. An excerpt from the pseudocode for this procedure is presented in Algorithm~\ref{alg:proof:search}. Given a finite set $\Gamma\subseteq\Langlog{\Logic}$ and a formula $\varphi\in\Langlog{\Logic}$ as input, our algorithm first initializes $\Tree$ as the proof search tree consisting only of the root node $\lab{\pi}{\Gamma}\seq\lab{\pi}{\varphi}$, where $\pi$ is a label with $|\pi| = 2^{|\Prop|}$ (line~\ref{alg:init}). Afterwards, the procedure successively extends $\Tree$ by applying the rules of $\GTsystem{\Logic}$ root-first in all possible ways to the topmost sequents in $\Tree$. The corresponding while-loops for $\Rrule{\rightarrow}$, $\Lrule{\rightarrow}$, $\Rrule{\tensor}$, $\Lrule{=}$ are displayed in lines~\ref{alg:start:expansion}--\ref{alg:constancy:left:end} (the while-loops for the other rules are similar). In applications of the cumulative rules $\Latom$, $\Lrule{\bot}$, $\Lrule{\rightarrow}$, $\Rrule{\tensor}$, $\Lrule{=}$, we use a suitable loop-check\-ing mechanism to prevent the algorithm from producing infinitely many copies of ``equivalent'' rule applications. For example, before performing a root-first application of $\Rrule{\tensor}$ with principal formula $\lab{\sigma}{\psi\tensor\chi}$ for a split $(\tau_{1},\tau_{2})\in\Splits{\sigma}$ in a branch $B$, we first check whether $(\lab{\tau_{1}}{\psi})\in\Ddown$ or $(\lab{\tau_{2}}{\chi})\in\Ddown$. If one of these holds, the application is taken to be redundant and not performed in $B$. Otherwise, the rule is applied in the usual way (lines~\ref{alg:tensor:right:start}--\ref{alg:tensor:right:end}). The corresponding loop-checking conditions for $\Lrule{\rightarrow}$ and $\Lrule{=}$ are displayed in lines~\ref{alg:impl:left:start} and \ref{alg:constancy:left:start}, respectively. The expansion of $\Tree$ is stopped, if either $\Tree$ becomes closed or a saturated branch is found (lines~\ref{alg:end:condition:success}--\ref{alg:end:condition:failure}). In the first case, the procedure outputs ``success'', as $\Tree$ is now a derivation for the root sequent. In the second case, the output is ``failure'', since a saturated branch will always allow us to construct an interpretation that falsifies the root sequent (see Lemma~\ref{lem:properties:interpretation}).  

\begin{lemma}
\label{lemma:termination}
Our proof search procedure for $\GTsystem{\Logic}$ terminates on every input $(\Gamma,\varphi)$, where $\Gamma$ is finite. 
\end{lemma}

\begin{proof}
Suppose that the algorithm is executed on some finite input $(\Gamma,\varphi)$, so it successively constructs a proof search tree $\Tree$ whose root sequent is of the form $\lab{\pi}{\Gamma}\seq\lab{\pi}{\varphi}$, where $|\pi| = 2^{|\Prop|}$. By inspection of the rules from Figure~\ref{fig:rules:terminating:calculi}, we can see that each labelled formula $\lab{\sigma}{\psi}$ that occurs in $\Tree$ satisfies the following two conditions: (1) $\psi$ is a subformula of some formula in $\Gamma\cup\{\varphi\}$, and (2) $\sigma$ is a subset of the label $\pi$. Since there are only finitely many such labelled formulas (and since we do not allow redundant applications of cumulative rules), it follows that only a finite number of rule applications is possible in the tree. Hence, every branch becomes either closed or saturated after a finite number of steps.    
\end{proof}

We will now prove that, from a saturated branch $B$, one can always extract an interpretation $\IB$ such that all elements of $\Gdown$ are satisfied and all elements of $\Ddown$ are not satisfied by $\IB$. This shows that our proof search procedure is in fact \emph{exhaustive}: if the algorithm does not find a proof for the root sequent $\lab{\pi}{\Gamma}\seq\lab{\pi}{\varphi}$ of the search tree, then this sequent can always be falsified by some interpretation. 

\begin{definition}
Let $B$ be a saturated branch. The interpretation $\IB :\Vvar\rightarrow \Val$ is defined as follows: for every $p\in\Prop$ and every $v\in\Vvar$, we put $\IB(v)(p) := 1$, if $(\lab{v}{p})\in\Gdown$, and we put $\IB(v)(p) := 0$ otherwise.  
\end{definition}

\begin{lemma}\label{lem:properties:interpretation}
Let $B$ be a saturated branch. For all $\varphi\in\Langlog{\Logic}$ and all labels $\pi\subseteq\Vvar$, the following holds:
\begin{itemize}
\setlength{\itemsep}{0pt}
\setlength{\parskip}{0pt}
\item if $(\lab{\pi}{\varphi})\in\Gdown$, then $\IB(\pi)\sem\varphi$,
\item if $(\lab{\pi}{\varphi})\in\Ddown$, then $\IB(\pi)\not\sem\varphi$.
\end{itemize}
\end{lemma}

\begin{proof}
Both claims are proved simultaneously by induction on $\varphi$. We only consider the following cases.

\emph{Case 1:} $\varphi = p$ for some $p\in\Prop$. For the first part, assume that $(\lab{\pi}{p})\in\Gdown$ and let $u\in\IB(\pi)$ be an arbitrary valuation. By definition of $\IB(\pi)$, there must be some variable $v\in\pi$ such that $u = \IB(v)$. From $(\lab{\pi}{p})\in\Gdown$ and $v\in\pi$, it now follows $(\lab{v}{p})\in\Gdown$ by Definition~\ref{def:saturation}~(\ref{def:sat:i}). Hence, by definition of $\IB$, we obtain $\IB(v)(p) = 1$ and thus $u(p) = 1$. Since $u\in\IB(\pi)$ was arbitrary, this shows that $\IB(\pi)\sem p$.

For the second part, suppose that $(\lab{\pi}{p})\in\Ddown$. Since $B$ is open, we must have $\pi\neq\emptyset$ (otherwise, the topmost sequent of $B$ would be an instance of $\Ratom$ with zero premises and therefore an initial sequent). Thus, by Definition~\ref{def:saturation}~(\ref{def:sat:ii}), there is some $v\in\pi$ such that $(\lab{v}{p})\in\Ddown$. As a consequence, we must have $(\lab{v}{p})\notin\Gdown$: otherwise, $B$ would contain an instance of $\ax$, which contradicts the assumption that $B$ is open. By definition of $\IB$, it now follows $\IB(v)(p) = 0$. Since $\IB(v)\in\IB(\pi)$, this yields $\IB(\pi)\not\sem p$.

\looseness=-1 \emph{Case 2:} $\varphi = \psi\rightarrow\chi$. For the first part, suppose for a contradiction that $(\lab{\pi}{\psi\rightarrow\chi})\in\Gdown$ and $\IB(\pi)\not\sem\psi\rightarrow\chi$, i.e., there is some team $T\subseteq \IB(\pi)$ with $T\sem\psi$ and $T\not\sem\chi$. Since $\IB(\pi)$ is finite, we can write $T$ in the form $T = \{\IB(v_{1}),\ldots,\IB(v_{n})\}$ for some variables $v_{1},\ldots,v_{n}\in\pi$ with $n\geq 0$. Let $\sigma$ be the label $\sigma := \{v_{1},\ldots,v_{n}\}$. We then have $T=\IB(\sigma)$ and $\sigma\subseteq\pi$. From $(\lab{\pi}{\psi\rightarrow\chi})\in\Gdown$ and $\sigma\subseteq\pi$, we obtain $(\lab{\sigma}{\psi})\in\Ddown$ or $(\lab{\sigma}{\chi})\in\Gdown$ by Definition~\ref{def:saturation}~(\ref{def:sat:ix}). By the induction hypothesis, this yields $\IB(\sigma)\not\sem\psi$ or $\IB(\sigma)\sem\chi$, so it follows $T\not\sem\psi$ or $T\sem\chi$. But this is a contradiction to the fact that $T\sem\psi$ and $T\not\sem\chi$. 

For the second part, assume that $(\lab{\pi}{\psi\rightarrow\chi})\in\Ddown$. By Definition~\ref{def:saturation}~(\ref{def:sat:x}), we then have $(\lab{\sigma}{\psi})\in\Gdown$ and $(\lab{\sigma}{\chi})\in\Ddown$ for some $\sigma\subseteq\pi$. By the induction hypothesis, this implies $\IB(\sigma)\sem\psi$ and $\IB(\sigma)\not\sem\chi$. Moreover, since $\sigma\subseteq\pi$, we must have $\IB(\sigma)\subseteq\IB(\pi)$. Hence, it follows $\IB(\pi)\not\sem\psi\rightarrow\chi$. 

\emph{Case 3:} $\varphi = \psi\tensor\chi$. For the first part, suppose $(\lab{\pi}{\psi\tensor\chi})\in\Gdown$. By Definition~\ref{def:saturation}~(\ref{def:sat:xi}), we then have $(\lab{\sigma}{\psi})\in\Gdown$ and $(\lab{\tau}{\chi})\in\Gdown$ for some $\sigma$, $\tau$ with $\pi = \sigma\cup\tau$. By the induction hypothesis, this yields $\IB(\sigma)\sem\psi$ and $\IB(\tau)\sem\chi$. And since $\pi = \sigma\cup\tau$, we have $\IB(\pi) = \IB(\sigma)\cup\IB(\tau)$, so $\IB(\pi)\sem\psi\tensor\chi$. 

For the second part, assume $(\lab{\pi}{\psi\tensor\chi})\in\Ddown$. Moreover, suppose for a contradiction that $\IB(\pi)\sem\psi\tensor\chi$, i.e., there are teams $S$ and $T$ with $\IB(\pi) = S\cup T$ such that $S\sem\psi$ and $T\sem\chi$. Since $\IB(\pi) = S\cup T$, there must be some labels $\sigma$ and $\tau$ with $\pi = \sigma\cup\tau$ such that $S=\IB(\sigma)$ and $T=\IB(\tau)$. For these labels, we have $\IB(\sigma)\sem\psi$ and $\IB(\tau)\sem\chi$. Moreover, since $(\lab{\pi}{\psi\tensor\chi})\in\Ddown$ and $\pi = \sigma\cup\tau$, it follows $(\lab{\sigma}{\psi})\in\Ddown$ or $(\lab{\tau}{\chi})\in\Ddown$ by Definition~\ref{def:saturation}~(\ref{def:sat:xii}). By the induction hypothesis, this yields $\IB(\sigma)\not\sem\psi$ or $\IB(\tau)\not\sem\chi$. But this is a contradiction to the fact that $\IB(\sigma)\sem\psi$ and $\IB(\tau)\sem\chi$. Therefore, $\IB(\pi)\not\sem\psi\tensor\chi$.

\emph{Case 4:} $\varphi$ is of the form $\dep{p}$. For the first part, suppose for a contradiction that $(\lab{\pi}{\dep{p}})\in\Gdown$ and $\IB(\pi)\not\sem\dep{p}$, i.e., there are valuations $u_{1},u_{2}\in\IB(\pi)$ such that $u_{1}(p)\neq u_{2}(p)$. Since $u_{1},u_{2}\in\IB(\pi)$, there must be variables $v_{1},v_{2}\in\pi$ with $u_{1} = \IB(v_{1})$ and $u_{2} = \IB(v_{2})$, so it follows $\IB(v_{1})(p) \neq \IB(v_{2})(p)$. Since $(\lab{\pi}{\dep{p}})\in\Gdown$ and $v_{1},v_{2}\in\pi$, we know by Definition~\ref{def:saturation}~(\ref{def:sat:xiii}) that one of the following holds: (1) $(\lab{v_{1}}{p})\in\Gdown$ and $(\lab{v_{2}}{p})\in\Gdown$, or (2) $(\lab{v_{1}}{p})\in\Ddown$ and $(\lab{v_{2}}{p})\in\Ddown$. In either case, it follows $\IB(v_{1})(p) = \IB(v_{2})(p)$ by the induction hypothesis. But this is a contradiction to $\IB(v_{1})(p) \neq \IB(v_{2})(p)$.

For the second part, assume $(\lab{\pi}{\dep{p}})\in\Ddown$. Since $B$ is open, we must have $\pi\neq\emptyset$ (otherwise, the topmost sequent of $B$ would be an instance of $\Rrule{=}$ with zero premises). Thus, by Definition~\ref{def:saturation}~(\ref{def:sat:xiv}), we have $(\lab{u}{p})\in\Gdown$ and $(\lab{v}{p})\in\Ddown$ for some $u,v\in\pi$. By the induction hypothesis, this yields $\IB(u)(p) \neq \IB(v)(p)$. And since $u,v\in\pi$, we have $\IB(u),\IB(v)\in\IB(\pi)$, so it follows $\IB(\pi)\not\sem\dep{p}$.
\end{proof}

\begin{theorem}[Completeness]
The calculus $\GTsystem{\Logic}$ is sound and complete with respect to $\Logic$. That is, for any finite subset $\Gamma\cup\{\varphi\}\subseteq\Langlog{\Logic}$, it is the case that $\Gamma\provGT\varphi$ holds in $\GTsystem{\Logic}$ iff $\Gamma\semL\varphi$ holds in $\Logic$.    
\end{theorem}

\begin{proof}
The soundness of $\GTsystem{\Logic}$ has been established in Lemma~\ref{lem:soundness:terminating:system}. For the completeness direction, let $\Gamma\cup\{\varphi\}\subseteq\Langlog{\Logic}$ be finite and suppose that $\Gamma\notprovGT\varphi$. Let $\Tree$ be the tree constructed by our proof search algorithm for the input $(\Gamma,\varphi)$, so the root of $\Tree$ is of the form $\lab{\pi}{\Gamma}\seq\lab{\pi}{\varphi}$, where $|\pi| = 2^{|\Prop|}$. Since $\Gamma\notprovGT\varphi$, we know that $\Tree$ cannot be closed, so it must have a saturated branch $B$. Let $\IB$ be the interpretation determined by $B$. By Lemma~\ref{lem:properties:interpretation}, all elements of $\Gdown$ are satisfied and all elements of $\Ddown$ are not satisfied by $\IB$. Thus, in particular, we have $\IB(\pi)\sem\psi$ for all $\psi\in\Gamma$ and $\IB(\pi)\not\sem\varphi$. Therefore, $\Gamma\not\semL\varphi$.    
\end{proof}

\section{Conclusion}\label{sec:conclusion}

We presented two kinds of labelled sequent calculi for a family of team-based propositional logics. We first introduced the systems $\Gsystem{\Logic}$, whose labels are strings of symbols representing arbitrary teams, their unions and intersections. The order rules of these calculi explicitly reflect the semantic reasoning on teams. We then defined the systems $\GTsystem{\Logic}$, whose labels are finite sets of variables. This simpler framework made it possible to provide a terminating proof search procedure for the calculi.

The proof systems $\GTsystem{\Logic}$ have the clear advantage of inducing a decision algorithm for the logics under consideration. The calculi $\Gsystem{\Logic}$, on the other hand, seem to have two other advantages. First, each rule of $\Gsystem{\Logic}$ has a fixed number of premises, whereas in $\GTsystem{\Logic}$, the number of premises of certain rules depends on the size of the label associated with the principal formula. Thus, at least for the tensor-free systems, proofs in $\Gsystem{\Logic}$ can be significantly shorter than proofs in $\GTsystem{\Logic}$. Second, in the calculi $\Gsystem{\Logic}$, one may use variables to reason about arbitrary teams (or arbitrary subsets of a given team), which more closely mirrors the actual mathematical reasoning one would employ to show that a certain formula is satisfied by a team or that all formulas of a certain form are valid in the corresponding logic. 

For future work, we intend to further simplify our calculi by circumventing the need for intersection labels, and by extending our proof systems for $\InqBtensor$ and $\PIDtensor$ to languages with infinitely many atomic propositions. Moreover, we plan to develop labelled calculi for broader families of team-based logics. A primary goal is to find proof systems for downward-closed team logics \emph{without} the intuitionistic implication $\rightarrow$, most notably for {propositional dependence logic} \cite{yang:vaananen:2016}. For such logics, only natural deduction systems \cite{yang:vaananen:2016} and deep-inference-style sequent calculi \cite{anttila2025deep} are known. We also intend to investigate team logics with the union closure property \cite{Yan2022,yang:vaananen:2017}, for which no sequent-style proof systems have been developed up to now. Another natural step would be to extend our calculi to logics based on modal or intuitionistic variants of team semantics \cite{Vaa2008,ciardelli:iemhoff:yang:2020}. Finally, we plan to use our calculi to establish meta-theoretical properties, such as syntactic interpolation, for the logics under consideration.    
\vspace{.5\baselineskip} 

\noindent\textbf{Acknowledgements.} We are grateful to the anonymous reviewers for their valuable feedback and helpful suggestions for improvement. Fausto Barbero was supported by the Research Council of Finland under grant number 349803. Valentin M{\"u}ller was supported by the Swiss National Science Foundation (SNSF) under grant number 10000440 (Epistemic Group Attitudes).

\appendix
\section{Proof of Theorem~\ref{thm:struct}}
\label{app:proof:struct:properties}

\emph{\ref{thm:struct:i}.~The weakening rules are hp-admissible in $\Gsystem{\Logic}$.} This is proved by an easy induction on the height of a derivation for the premise $\Gamma\seq\Delta$ of the respective rule. More details can be found in~\cite[Prop.~3.2.7]{muller:msc:2023}.\vspace{.5\baselineskip} 

\noindent\emph{\ref{thm:struct:ii}.~All rules of $\Gsystem{\Logic}$ are hp-invertible.} The hp-invertibility of the basic and special order rules follows immediately from the hp-admissibility of weakening. The same holds for the rules $\Latom$, $\Lrule{\bot}$, $\Lrule{\rightarrow}$, $\Rrule{\tensor}$, $\Lrule{=}$. For all other rules, we reason by induction on the height of a derivation for the conclusion of the rule. We only show the invertibility of the rule $\Lrule{\tensor}$ (the other cases are similar). Let $\mathcal{D}$ be a derivation for the sequent $\lab{\A}{\varphi\tensor\psi}, \extseq$, and let $n$ be the height of $\mathcal{D}$.  

If $n = 0$, then $\lab{\A}{\varphi \tensor \psi}, \extseq$ must be an axiom. Consequently, $\lab{\A}{\varphi \tensor \psi}$ cannot be principal, as it is not atomic. Thus, for all variables $x, y \in \Avar$, the sequent $\relat{\A}{x \cup y}, \lab{x}{\varphi}, \lab{y}{\psi}, \extseq$ is also an axiom. 

If $n \geq 1 $, we distinguish cases according to the last rule $\rulegen$ applied in $\mathcal{D}$. If $\lab{\pi}{\varphi \tensor \psi}$ is not principal in~$\rulegen$, we apply the induction hypothesis to the premises of $\rulegen$ (possibly with a renaming of eigenvariables), and apply $\rulegen$ again. If $\lab{\pi}{\varphi \tensor \psi}$ is principal, then $\rulegen$ is an application of $\Lrule{\tensor}$. The premise of $\Lrule{\tensor}$ is of the form $\relat{\pi}{x' \cup y'}, \lab{x'}{\varphi}, \lab{y'}{\psi}, \extseq$, where $x', y' \in \Avar$ are fresh. This sequent is derivable by a derivation of height $n-1$. By hp-substitution, we obtain a derivation of height $\leq n$ for $\relat{\pi}{x \cup y}, \lab{x}{\varphi}, \lab{y}{\psi}, \extseq$.\vspace{.5\baselineskip}

\noindent\emph{\ref{thm:struct:iii}.~The contraction rules are hp-admissible in $\Gsystem{\Logic}$.} The hp-admissibility of both contraction rules is proved simultaneously, by induction on the height of a derivation $\mathcal{D}$ for the premise, and by distinguishing cases according to the last rule applied in $\mathcal D$. Let $\Gamma\seq\Delta$ be a sequent that contains a duplication of some relational atom or labelled formula $E$ in the antecedent or in the succedent. Let $\mathcal{D}$ be a derivation for $\Gamma\seq\Delta$ in $\Gsystem{\Logic}$ and let $n$ be the height of $\mathcal{D}$. Using induction on $n$, we show that also the contracted version of $\Gamma\seq\Delta$ (in which the two occurrences of $E$ are replaced by a single one) is derivable by a proof tree of height $\leq n$. If $\mathcal{D}$ is of height $n=0$, then $\Gamma\seq\Delta$ is an axiom, so the contracted version of $\Gamma\seq\Delta$ is also an axiom and therefore derivable by a proof tree of height $n=0$. 

\looseness=-1 Suppose now that $n\geq 1$. If $E$ is not principal in the last rule $\rulegen$ applied in $\mathcal{D}$, we apply the induction hypothesis to the premises of~$\rulegen$, and then apply $\rulegen$ again. And if $E$ is principal in $\rulegen$, we distinguish two cases.

\emph{Case 1:} $E$ is a relational atom. If only one of the two occurrences of $E$ is principal in the last rule application~$\rulegen$ in $\mathcal{D}$, then we use the same argument as above. And if both occurrences of $E$ are principal in $\rulegen$, then we use the \emph{closure condition} introduced in Section~\ref{sec:calculi}.

\emph{Case 2:} $E$ is a labelled formula. In this case, we use the hp-invertibility of the logical rules of $\Gsystem{\Logic}$. 
For example, if $\mathcal{D}$ ends with an application of $\Lrule{\tensor}$ for which $E$ is principal, the premise of the rule is of the form $\relat{\pi}{x \cup y},\lab x \varphi, \lab{y}{\psi},\lab{\pi}{\varphi \tensor \psi}, \Sigma\seq\Delta$, for some fresh $x, y \in \Avar$. The height of the derivation of this sequent is $n-1$. By hp-invertibility and hp-substitution, we obtain a derivation of height $\leq n-1$ for the sequent $\relat{\pi}{x \cup y}, \relat{\pi}{x \cup y},\lab x \varphi,\lab x \varphi, \lab{y}{\psi},\lab{y}{\psi}, \Sigma\seq\Delta$. Using three applications of the inductive hypothesis (which do not increase the height) we obtain a derivation of height $\leq n-1$ for $\relat{\pi}{x \cup y},\lab x \varphi, \lab{y}{\psi}, \Sigma\seq\Delta$. An application of $\Lrule{\tensor}$ now yields a derivation of height $\leq n$ for $\lab{\pi}{\varphi\tensor \psi}, \Sigma\seq\Delta$.\vspace{.5\baselineskip}

\noindent\emph{\ref{thm:struct:iv}.~The cut rule is admissible in $\Gsystem{\Logic}$.} The proof works in the standard way (see, e.g., \cite[Theorem~2.4.3]{negri:plato:2001}) and proceeds by induction on two parameters: the \emph{height} of a cut rule application and the \emph{rank} of the cut formula, which is a measure of the complexity of labelled formulas.\footnote{By the \emph{cut formula}, we mean the labelled formula $\lab{\pi}{\varphi}$ displayed in the presentation of the cut rule in Figure~\ref{fig:structural:rules}.} The \emph{height} of a cut rule application is the sum of the heights of the two subderivations ending with the premises of the cut. Our measure of the complexity of a labelled formula $\lab{\pi}{\varphi}$ takes into account both the complexity of $\varphi$ and the complexity of the label $\pi$ (this is needed to deal with cuts involving the rules $\Latom$, $\Ratom$, $\Lrule{\bot}$, $\Rrule{\bot}$). 

Formally, the \emph{degree of a label} $\pi\in\LABint$, notation $\dg{\pi}$, is defined by setting $\dg{\pi} := 0$, if $\pi\in\Svar$, and $\dg{\pi} := 1$, if $\pi\notin\Svar$. The \emph{degree of a formula} $\varphi \in \Langlog{\Logic}$, notation $\dg{\varphi}$, is defined to be the number of occurrences of the connectives $\bot,\wedge,\globdis,\rightarrow,\dep{\cdot},\tensor$ in $\varphi$. The \emph{rank of a labelled formula} is the pair of natural numbers given by $\rk{\lab{\pi}{\varphi}} := (\dg{\varphi},\dg{\pi})$. As in \cite{muller:2024}, we assume that ranks of labelled formulas are ordered \emph{lexicographically}. That is, we write $\rk{\lab{\pi}{\varphi}}< \rk{\lab{\sigma}{\psi}}$, if we either have $\dg{\varphi} < \dg{\psi}$, or we have both $\dg{\varphi} = \dg{\psi}$ and $\dg{\pi} < \dg{\sigma}$. It is easy to verify that, if $\pi$ is a label with $\pi\notin\Svar$, then $\rk{\lab{v}{\varphi}}< \rk{\lab{\pi}{\varphi}}$ for all variables $v\in\Svar$ and all formulas $\varphi\in\Langlog{\Logic}$. 

The cut-admissibility proof proceeds by a main induction on the rank of the cut formula and a subinduction on the height of the cut. We distinguish three main cases, depending on the two premises of the cut: \emph{(i)} one of the two premises of the cut is an axiom; \emph{(ii)} neither of the two premises is an axiom and the cut formula is not principal on one side; and \emph{(iii)} neither of the two premises of the cut is an axiom and the cut formula is principal on both sides. We show only two subcases of case \emph{(iii)}; the other cases follow quite standardly, and a discussion for several of them can be found in \cite[Theorem~3.2.13]{muller:msc:2023}.

First, suppose that the cut formula is a propositional atom. In this case, the cut must be of the form
\begin{center}
\myfontsize
\AxiomC{$\mathcal{D}_1$}
\noLine
\UnaryInfC{$\relat{v}{\pi}, \Gamma\seq\Delta, \lab{v}{p}$}
\LeftLabel{$\Ratom$}
\UnaryInfC{$\Gamma\seq\Delta,\lab{\pi}{p}$}
\AxiomC{$\mathcal{D}_2$}
\noLine
\UnaryInfC{$\lab{u}{p}, \relat{u}{\pi}, \lab{\pi}{p}, \Pi\seq\Sigma$}
\LeftLabel{$\Latom$}
\UnaryInfC{$\relat{u}{\pi}, \lab{\pi}{p}, \Pi\seq\Sigma$}
\LeftLabel{$\cut$}
\BinaryInfC{$\relat{u}{\pi}, \Gamma,\Pi\seq\Delta,\Sigma$}
\DisplayProof
\end{center}
where $\pi$ is a label with $\pi\notin\Svar$ and $v$ is a fresh variable from $\Svar$. Using the hp-admissibility of substitution and contraction, this application of the cut rule is converted into
\begin{center}
\myfontsize
\AxiomC{$\mathcal{D}_1$}
\noLine
\UnaryInfC{$\relat{v}{\pi}, \Gamma\seq\Delta, \lab{v}{p}$}
\LeftLabel{$(u/v)$}
\UnaryInfC{$\relat{u}{\pi}, \Gamma\seq\Delta, \lab{u}{p}$}
\AxiomC{$\mathcal{D}_1$}
\noLine
\UnaryInfC{$\relat{v}{\pi}, \Gamma\seq\Delta, \lab{v}{p}$}
\LeftLabel{$\Ratom$}
\UnaryInfC{$\Gamma\seq\Delta,\lab{\pi}{p}$}
\AxiomC{$\mathcal{D}_2$}
\noLine
\UnaryInfC{$\lab{u}{p}, \relat{u}{\pi}, \lab{\pi}{p}, \Pi\seq\Sigma$}
\LeftLabel{$\cut$}
\BinaryInfC{$\lab{u}{p}, \relat{u}{\pi}, \Gamma,\Pi\seq\Delta,\Sigma$}
\LeftLabel{$\cut$}
\BinaryInfC{$\relat{u}{\pi}, \relat{u}{\pi}, \Gamma,\Gamma,\Pi\seq\Delta,\Delta,\Sigma$}
\LeftLabel{$\CL,\CR$}
\UnaryInfC{$\relat{u}{\pi}, \Gamma,\Pi\seq\Delta,\Sigma$}
\DisplayProof
\end{center}
where $\CL, \CR$ denote multiple applications of the contraction rules. The new cut on $\lab{\pi}{p}$ is of smaller height than the original one, and the cut on $\lab{u}{p}$ is of smaller rank (since $u\in\Svar$ and $\pi\notin\Svar$). Next, suppose that the cut formula is of the form  $\lab{\pi}{\varphi\tensor\psi}$. Our derivation then has the following shape:
\begin{center}
\myfontsize
\AxiomC{$\mathcal{D}'_1$}
\noLine
\UnaryInfC{$\relat{\pi}{\sigma\cup\tau}, \extseq, \lab{\pi}{\varphi\tensor\psi}, \lab{\sigma}{\varphi}$}
\AxiomC{$\mathcal{D}''_1$}
\noLine
\UnaryInfC{$\relat{\pi}{\sigma\cup\tau}, \extseq, \lab{\pi}{\varphi\tensor\psi}, \lab{\tau}{\psi}$}
\LeftLabel{$\Rrule{\tensor}$}
\BinaryInfC{$\relat{\pi}{\sigma\cup\tau}, \extseq, \lab{\pi}{\varphi\tensor\psi}$}
\AxiomC{$\mathcal{D}_2$}
\noLine
\UnaryInfC{$\relat{\pi}{x\cup y}, \lab{x}{\varphi}, \lab{y}{\psi}, \Pi \seq \Sigma $}
\LeftLabel{ $\Lrule{\tensor}$}
\UnaryInfC{$\lab{\pi}{\varphi\tensor\psi}, \Pi \seq \Sigma  $}
\BinaryInfC{$\relat{\pi}{\sigma \cup \tau}, \Gamma, \Pi \seq \Delta,\Sigma  $}
\DisplayProof 
\end{center}
Without loss of generality, we may assume that the eigenvariables $x$ and $y$ do not occur in the labels $\sigma$ and $\tau$ (if this condition is not satisfied, we simply apply the hp-admissibility of substitution). We construct the following derivation, where $\mathcal{E}$, the subderivation for the sequent $\relat{\pi}{\sigma\cup\tau}, \Gamma, \Pi \seq \Delta, \Sigma, \lab{\tau}{\psi}$, is constructed in a similar way as the subderivation for $\relat{\pi}{\sigma\cup\tau}, \Gamma, \Pi \seq \Delta, \Sigma, \lab{\sigma}{\varphi}$ displayed below:  
\begin{center}
\myfontsize
\def\defaultHypSeparation{\hskip .1in}
\AxiomC{$\mathcal{E}$}
\noLine
\UnaryInfC{$\relat{\pi}{\sigma\cup\tau}, \Gamma, \Pi \seq \Delta, \Sigma, \lab{\tau}{\psi}$}
\AxiomC{$\mathcal{D}'_1$}
\noLine
\UnaryInfC{$\relat{\pi}{\sigma\cup\tau}, \extseq, \lab{\pi}{\varphi\tensor\psi}, \lab{\sigma}{\varphi}$}
\AxiomC{$\mathcal{D}_2$}
\noLine
\UnaryInfC{$\relat{\pi}{x\cup y}, \lab{x}{\varphi}, \lab{y}{\psi}, \Pi \seq \Sigma $}
\LeftLabel{ $\Lrule{\tensor}$}
\UnaryInfC{$\lab{\pi}{\varphi\tensor\psi}, \Pi \seq \Sigma  $}
\LeftLabel{$\cut$}
\BinaryInfC{$ \relat{\pi}{\sigma\cup\tau}, \Gamma, \Pi \seq \Delta, \Sigma, \lab{\sigma}{\varphi}$}
\AxiomC{$\mathcal{D}_2$}
\noLine
\UnaryInfC{$\relat{\pi}{x\cup y}, \lab{x}{\varphi}, \lab{y}{\psi}, \Pi \seq \Sigma $}
\LeftLabel{$(\tau/y)$}
\UnaryInfC{$\relat{\pi}{x\cup \tau}, \lab{x}{\varphi}, \lab{\tau}{\psi}, \Pi \seq \Sigma $}
\LeftLabel{$(\sigma/x)$}
\UnaryInfC{$\relat{\pi}{\sigma\cup \tau}, \lab{\sigma}{\varphi}, \lab{\tau}{\psi}, \Pi \seq \Sigma $}
\LeftLabel{$\cut$}
\BinaryInfC{$ \relat{\pi}{\sigma\cup \tau},\relat{\pi}{\sigma\cup \tau}, \lab{\tau}{\psi}, \Gamma,  \Pi, \Pi \seq \Delta, \Sigma, \Sigma$}
\def\defaultHypSeparation{\hskip -1cm}
\LeftLabel{$\cut$}
\BinaryInfC{$ \relat{\pi}{\sigma\cup \tau},\relat{\pi}{\sigma\cup \tau},\relat{\pi}{\sigma\cup \tau},  \Gamma,\Gamma,  \Pi, \Pi, \Pi \seq \Delta,\Delta, \Sigma,\Sigma, \Sigma$}
\def\defaultHypSeparation{\hskip .2in}
\LeftLabel{$\CL,\CR$}
\UnaryInfC{$ \relat{\pi}{\sigma\cup \tau},\Gamma,  \Pi \seq \Delta, \Sigma$}
\DisplayProof
\end{center}
The cut on $\lab{\pi}{\varphi \tensor \psi}$ has smaller height than the original one, and is thus justified by the inductive hypothesis. Another such cut is contained in the subderivation $\mathcal{E}$. The cut on $\lab{\sigma}{\varphi}$ and the one on $\lab{\tau}{\psi}$ have smaller rank than the original cut, and are also justified by the inductive hypothesis.\qed

\section{Proof of Lemma~\ref{lem:auxiliary:rules}}\label{app:proof:aux:rules}

Let $\Logic$ be any of the logics from Table~\ref{tab:team:logics} and let $\Gsystem{\Logic}$ be the corresponding labelled calculus from Figure~\ref{fig:labelled calculi}. We only prove the admissibility of the rule $\auxi$ (the admissibility proof for $\auxii$ works in the same way). Suppose that the premise $\relat{\pi}{(\sigma\cup\tau)\cup(\sigma'\cup\tau')},\Gamma\seq\Delta$ is derivable in $\Gsystem{\Logic}$. Using the admissibility of weakening, we can then derive the conclusion of $\auxi$ in the following way: 
\begin{center}
\myfontsize
\AxiomC{$\relat{\pi}{(\sigma\cup\tau)\cup(\sigma'\cup\tau')},\Gamma\seq\Delta$}
\LeftLabel{$\WL$}
\UnaryInfC{$\ldots,\relat{\pi}{(\sigma\cup\tau)\cup(\sigma'\cup\tau')},\Gamma\seq\Delta$}
\LeftLabel{$\tr$}
\UnaryInfC{$\ldots,\relat{\sigma\cup\sigma'}{(\sigma\cup\tau)\cup(\sigma'\cup\tau')},\relat{\pi}{\sigma\cup\sigma'},\Gamma\seq\Delta$}
\LeftLabel{$\ul$}
\UnaryInfC{$\ldots,\relat{\sigma}{(\sigma\cup\tau)\cup(\sigma'\cup\tau')},\relat{\sigma'}{(\sigma\cup\tau)\cup(\sigma'\cup\tau')},\relat{\pi}{\sigma\cup\sigma'},\Gamma\seq\Delta$}
\LeftLabel{$\tr$}
\UnaryInfC{$\ldots,\relat{\sigma}{\sigma\cup\tau},\relat{\sigma\cup\tau}{(\sigma\cup\tau)\cup(\sigma'\cup\tau')}, \relat{\sigma'}{(\sigma\cup\tau)\cup(\sigma'\cup\tau')},\relat{\pi}{\sigma\cup\sigma'},\Gamma\seq\Delta$}
\LeftLabel{$\ur$}
\UnaryInfC{$\ldots,\relat{\sigma\cup\tau}{(\sigma\cup\tau)\cup(\sigma'\cup\tau')}, \relat{\sigma'}{(\sigma\cup\tau)\cup(\sigma'\cup\tau')},\relat{\pi}{\sigma\cup\sigma'},\Gamma\seq\Delta$}
\LeftLabel{$\ur$}
\UnaryInfC{$\ldots, \relat{\sigma'}{(\sigma\cup\tau)\cup(\sigma'\cup\tau')},\relat{\pi}{\sigma\cup\sigma'},\Gamma\seq\Delta$}
\LeftLabel{$\tr$}
\UnaryInfC{$\relat{\sigma'}{\sigma'\cup\tau'}, \relat{\sigma'\cup\tau'}{(\sigma\cup\tau)\cup(\sigma'\cup\tau')},\relat{\pi}{\sigma\cup\sigma'},\Gamma\seq\Delta$}
\LeftLabel{$\ur$}
\UnaryInfC{$\relat{\sigma'\cup\tau'}{(\sigma\cup\tau)\cup(\sigma'\cup\tau')},\relat{\pi}{\sigma\cup\sigma'},\Gamma\seq\Delta$}
\LeftLabel{$\ur$}
\UnaryInfC{$\relat{\pi}{\sigma\cup\sigma'},\Gamma\seq\Delta$}
\DisplayProof
\end{center}

\section{Supplement to the Proof of Lemma~\ref{lem:complete:wrt:hilbert}}\label{app:proof:complete:wrt:hilbert}

We need to show that all axioms of the Hilbert system $\Hsystem{\Logic}$ are provable in the labelled calculus $\Gsystem{\Logic}$ (see Definition~\ref{def:provability}). For the sake of brevity, we only consider the case $\Logic=\PIDtensor$ (the other cases are treated in the same way). Showing the provability of the schemes $\AxA{1}$--$\AxA{7}$, $\AxIntro$, $\AxDis$, $\AxCons$ is straightforward. The scheme $\AxSplit$ can be derived in the same way as in \cite[Lemma~3.3.5]{muller:msc:2023} and \cite[p.~137]{muller:2024}. And for the schemes $\AxDN$, $\AxCom$, $\AxMon$, $\AxElim$, we can construct the following derivations in $\GPIDtensor$: 
\begin{center}
\myfontsize
\AxiomC{Lemma~\ref{lem:gen:ax}~(\ref{gen:ax:ii})}
\dashedLine
\UnaryInfC{$\relat{z}{\emptyset}, \ldots \seq \lab{v}{\alpha}, \lab{z}{\bot}$}
\AxiomC{Lemma~\ref{lem:gen:ax}~(\ref{gen:ax:i})}
\dashedLine
\UnaryInfC{$\relat{v}{z}, \ldots, \lab{z}{\alpha} \seq \lab{v}{\alpha}, \lab{z}{\bot}$}
\LeftLabel{$\sg$}
\BinaryInfC{$\relat{z}{v}, \relat{v}{y}, \relat{y}{x}, \lab{y}{\neg\neg\alpha}, \lab{z}{\alpha} \seq \lab{v}{\alpha}, \lab{z}{\bot}$}
\LeftLabel{$\Rrule{\rightarrow}$}
\UnaryInfC{$\relat{v}{y}, \relat{y}{x}, \lab{y}{\neg\neg\alpha} \seq \lab{v}{\alpha}, \lab{v}{\neg\alpha}$\vphantom{$\bot$}}
\AxiomC{}
\LeftLabel{$\axBot$}
\UnaryInfC{$\lab{v}{\bot}, \ldots \seq \lab{v}{\alpha}$}
\LeftLabel{$\Lrule{\rightarrow}$}
\BinaryInfC{$\relat{v}{y}, \relat{y}{x}, \lab{y}{\neg\neg\alpha} \seq \lab{v}{\alpha}$\vphantom{$\bot$}}
\LeftLabel{$\FL$}
\UnaryInfC{$\relat{y}{x}, \lab{y}{\neg\neg\alpha} \seq \lab{y}{\alpha}$\vphantom{$\bot$}}
\LeftLabel{$\Rrule{\rightarrow}$}
\UnaryInfC{$\seq \lab{x}{\neg\neg\alpha\rightarrow \alpha}$\vphantom{$\bot$}}
\DisplayProof
\\[.5\baselineskip]
\AxiomC{Lemma~\ref{lem:gen:ax}~(\ref{gen:ax:iv})}
\dashedLine
\UnaryInfC{$\ldots, \lab{z_2}{\psi} \seq \ldots, \lab{z_2}{\psi}$\vphantom{$()$}}
\AxiomC{Lemma~\ref{lem:gen:ax}~(\ref{gen:ax:iv})}
\dashedLine
\UnaryInfC{$\ldots, \lab{z_1}{\varphi} \seq \ldots, \lab{z_1}{\varphi}$\vphantom{$()$}}
\LeftLabel{$\Rrule{\tensor}$}
\BinaryInfC{$\ldots, \relat{y}{z_2\cup z_1}, \lab{z_1}{\varphi}, \lab{z_2}{\psi} \seq \lab{y}{\psi\tensor\varphi}$\vphantom{$()$}}
\LeftLabel{$\tr$}
\UnaryInfC{$\ldots, \relat{z_1\cup z_2}{z_2\cup z_1}, \relat{y}{z_1\cup z_2}, \lab{z_1}{\varphi}, \lab{z_2}{\psi} \seq \lab{y}{\psi\tensor\varphi}$\vphantom{$()$}}
\LeftLabel{$\ul$}
\UnaryInfC{$\ldots, \relat{z_1}{z_2\cup z_1}, \relat{z_2}{z_2\cup z_1}, \relat{y}{z_1\cup z_2}, \lab{z_1}{\varphi}, \lab{z_2}{\psi} \seq \lab{y}{\psi\tensor\varphi}$\vphantom{$()$}}
\LeftLabel{$\ur$}
\UnaryInfC{$\relat{z_2}{z_2\cup z_1}, \relat{y}{z_1\cup z_2}, \relat{y}{x}, \lab{z_1}{\varphi}, \lab{z_2}{\psi} \seq \lab{y}{\psi\tensor\varphi}$\vphantom{$()$}}
\LeftLabel{$\ur$}
\UnaryInfC{$\relat{y}{z_1\cup z_2}, \relat{y}{x}, \lab{z_1}{\varphi}, \lab{z_2}{\psi} \seq \lab{y}{\psi\tensor\varphi}$\vphantom{$()$}}
\LeftLabel{$\Lrule{\tensor}$}
\UnaryInfC{$\relat{y}{x}, \lab{y}{\varphi\tensor\psi} \seq \lab{y}{\psi\tensor\varphi}$\vphantom{$()$}}
\LeftLabel{$\Rrule{\rightarrow}$}
\UnaryInfC{$\seq \lab{x}{(\varphi\tensor\psi)\rightarrow(\psi\tensor\varphi)}$}
\DisplayProof
\\[.5\baselineskip]
\AxiomC{Lemma~\ref{lem:gen:ax}~(\ref{gen:ax:i})}
\dashedLine
\UnaryInfC{$\relat{z\cap z_1}{z_1}, \ldots, \lab{z_1}{\varphi} \seq \ldots, \lab{z\cap z_1}{\varphi}$\vphantom{$()$}}
\LeftLabel{$\il$}
\UnaryInfC{$\ldots, \lab{z_1}{\varphi} \seq \ldots, \lab{z\cap z_1}{\varphi}$\vphantom{$()$}}
\AxiomC{Lemma~\ref{lem:gen:ax}~(\ref{gen:ax:iv})}
\dashedLine
\UnaryInfC{$\ldots, \lab{z\cap z_1}{\psi} \seq \ldots, \lab{z\cap z_1}{\psi}$\vphantom{$()$}}
\AxiomC{Lemma~\ref{lem:gen:ax}~(\ref{gen:ax:i})}
\dashedLine
\UnaryInfC{$\relat{z\cap z_2}{z_2}, \ldots, \lab{z_2}{\chi} \seq \ldots, \lab{z\cap z_2}{\chi}$\vphantom{$()$}}
\LeftLabel{$\il$}
\UnaryInfC{$\ldots, \lab{z_2}{\chi} \seq \ldots, \lab{z\cap z_2}{\chi}$\vphantom{$()$}}
\LeftLabel{$\Rrule{\tensor}$}
\BinaryInfC{$\ldots, \relat{z}{(z\cap z_1)\cup(z\cap z_2)}, \lab{z\cap z_1}{\psi}, \lab{z_2}{\chi} \seq \lab{z}{\psi\tensor\chi}$\vphantom{$()$}}
\LeftLabel{$\dis$}
\UnaryInfC{$\ldots, \relat{z}{z_1\cup z_2}, \lab{z\cap z_1}{\psi}, \lab{z_2}{\chi} \seq \lab{z}{\psi\tensor\chi}$\vphantom{$()$}}
\LeftLabel{$\Lrule{\rightarrow}$}
\BinaryInfC{$\ldots, \relat{z\cap z_1}{y}, \relat{z}{z_1\cup z_2}, \lab{y}{\varphi\rightarrow\psi}, \lab{z_1}{\varphi}, \lab{z_2}{\chi} \seq \lab{z}{\psi\tensor\chi}$\vphantom{$()$}}
\LeftLabel{$\tr$}
\UnaryInfC{$\relat{z\cap z_1}{z}, \relat{z}{z_1\cup z_2}, \relat{z}{y}, \relat{y}{x}, \lab{y}{\varphi\rightarrow\psi}, \lab{z_1}{\varphi}, \lab{z_2}{\chi} \seq \lab{z}{\psi\tensor\chi}$\vphantom{$()$}}
\LeftLabel{$\il$}
\UnaryInfC{$\relat{z}{z_1\cup z_2}, \relat{z}{y}, \relat{y}{x}, \lab{y}{\varphi\rightarrow\psi}, \lab{z_1}{\varphi}, \lab{z_2}{\chi} \seq \lab{z}{\psi\tensor\chi}$\vphantom{$()$}}
\LeftLabel{$\Lrule{\tensor}$}
\UnaryInfC{$\relat{z}{y}, \relat{y}{x}, \lab{y}{\varphi\rightarrow\psi}, \lab{z}{\varphi\tensor\chi} \seq \lab{z}{\psi\tensor\chi}$\vphantom{$()$}}
\LeftLabel{$\Rrule{\rightarrow}$}
\UnaryInfC{$\relat{y}{x}, \lab{y}{\varphi\rightarrow\psi} \seq \lab{y}{(\varphi\tensor\chi) \rightarrow (\psi\tensor\chi)}$}
\LeftLabel{$\Rrule{\rightarrow}$}
\UnaryInfC{$\seq \lab{x}{(\varphi\rightarrow\psi) \rightarrow ((\varphi\tensor\chi) \rightarrow (\psi\tensor\chi))}$}
\DisplayProof 
\\[.5\baselineskip]
\AxiomC{Lemma~\ref{lem:gen:ax}~(\ref{gen:ax:i})}
\dashedLine
\UnaryInfC{$\ldots, \relat{v}{x_1}, \lab{x_1}{\varphi} \seq \ldots, \lab{v}{\varphi}$\vphantom{$()$}}
\AxiomC{Lemma~\ref{lem:gen:ax}~(\ref{gen:ax:iv})}
\dashedLine
\UnaryInfC{$\ldots, \lab{v}{\alpha} \seq \lab{v}{\alpha}$\vphantom{$()$}}
\LeftLabel{$\Lrule{\rightarrow}$}
\BinaryInfC{$\ldots, \relat{v}{x_1}, \relat{v}{y}, \lab{y}{\varphi\rightarrow\alpha}, \lab{x_1}{\varphi} \seq \lab{v}{\alpha}$\vphantom{$()$}}
\AxiomC{Lemma~\ref{lem:gen:ax}~(\ref{gen:ax:i})}
\dashedLine
\UnaryInfC{$\ldots, \relat{v}{x_2}, \lab{x_2}{\psi} \seq \ldots, \lab{v}{\psi}$\vphantom{$()$}}
\AxiomC{Lemma~\ref{lem:gen:ax}~(\ref{gen:ax:iv})}
\dashedLine
\UnaryInfC{$\ldots, \lab{v}{\alpha} \seq \lab{v}{\alpha}$\vphantom{$()$}}
\LeftLabel{$\Lrule{\rightarrow}$}
\BinaryInfC{$\ldots, \relat{v}{x_2}, \relat{v}{z}, \lab{z}{\psi\rightarrow\alpha}, \lab{x_2}{\psi} \seq \lab{v}{\alpha}$\vphantom{$()$}}
\LeftLabel{$\cd$}
\BinaryInfC{$\ldots, \relat{v}{x_1\cup x_2}, \relat{v}{z}, \relat{v}{y}, \lab{y}{\varphi\rightarrow\alpha}, \lab{z}{\psi\rightarrow\alpha}, \lab{x_1}{\varphi}, \lab{x_2}{\psi} \seq \lab{v}{\alpha}$\vphantom{$()$}}
\LeftLabel{$\tr$}
\UnaryInfC{$\ldots, \relat{v}{x_1\cup x_2}, \relat{v}{z}, \relat{z}{y}, \lab{y}{\varphi\rightarrow\alpha}, \lab{z}{\psi\rightarrow\alpha}, \lab{x_1}{\varphi}, \lab{x_2}{\psi} \seq \lab{v}{\alpha}$\vphantom{$()$}}
\LeftLabel{$\tr$}
\UnaryInfC{$\ldots, \relat{v}{x_1\cup x_2}, \relat{v}{x_0}, \relat{x_0}{z}, \relat{z}{y}, \lab{y}{\varphi\rightarrow\alpha}, \lab{z}{\psi\rightarrow\alpha}, \lab{x_1}{\varphi}, \lab{x_2}{\psi} \seq \lab{v}{\alpha}$\vphantom{$()$}}
\LeftLabel{$\tr$}
\UnaryInfC{$\relat{x_0}{x_1\cup x_2}, \relat{v}{x_0}, \relat{x_0}{z}, \relat{z}{y}, \relat{y}{x}, \lab{y}{\varphi\rightarrow\alpha}, \lab{z}{\psi\rightarrow\alpha}, \lab{x_1}{\varphi}, \lab{x_2}{\psi} \seq \lab{v}{\alpha}$\vphantom{$()$}}
\LeftLabel{$\Lrule{\tensor}$}
\UnaryInfC{$\relat{v}{x_0}, \relat{x_0}{z}, \relat{z}{y}, \relat{y}{x}, \lab{y}{\varphi\rightarrow\alpha}, \lab{z}{\psi\rightarrow\alpha}, \lab{x_0}{\varphi\tensor\psi} \seq \lab{v}{\alpha}$\vphantom{$()$}}
\LeftLabel{$\FL$}
\UnaryInfC{$\relat{x_0}{z}, \relat{z}{y}, \relat{y}{x}, \lab{y}{\varphi\rightarrow\alpha}, \lab{z}{\psi\rightarrow\alpha}, \lab{x_0}{\varphi\tensor\psi} \seq \lab{x_0}{\alpha}$\vphantom{$()$}}
\LeftLabel{$\Rrule{\rightarrow}$}
\UnaryInfC{$\relat{z}{y}, \relat{y}{x}, \lab{y}{\varphi\rightarrow\alpha}, \lab{z}{\psi\rightarrow\alpha} \seq \lab{z}{(\varphi\tensor\psi) \rightarrow \alpha}$}
\LeftLabel{$\Rrule{\rightarrow}$}
\UnaryInfC{$\relat{y}{x}, \lab{y}{\varphi\rightarrow\alpha} \seq \lab{y}{(\psi\rightarrow\alpha) \rightarrow ((\varphi\tensor\psi) \rightarrow \alpha)}$}
\LeftLabel{$\Rrule{\rightarrow}$}
\UnaryInfC{$\seq \lab{x}{(\varphi\rightarrow\alpha) \rightarrow ((\psi\rightarrow\alpha) \rightarrow ((\varphi\tensor\psi) \rightarrow \alpha))}$}
\DisplayProof 
\end{center}
Note that, in the derivations for $\AxDN$ and $\AxElim$, we also use the admissible rule $\FL$ from Figure~\ref{fig:further:admissible:rules}.

\bibliographystyle{eptcs}
\bibliography{references}

\end{document}